\documentclass[12pt]{article}
\usepackage{amssymb,amsmath,color,natbib,graphicx,amsthm,
  setspace,sectsty,anysize,times,dsfont}

\usepackage{lscape,arydshln,relsize,rotating}
\usepackage[small]{caption}

\newtheorem{prop}{\sc Proposition}[section]

\marginsize{1.1in}{.9in}{.3in}{1.4in}

\newcommand{\dbl}{\setstretch{1.5}}
\newcommand{\sgl}{\setstretch{1.1}}

\newcommand{\bs}[1]{\boldsymbol{#1}}
\newcommand{\mc}[1]{\mathcal{#1}}
\newcommand{\mr}[1]{\mathrm{#1}}
\newcommand{\bm}[1]{\mathbf{#1}}
\newcommand{\ds}[1]{\mathds{#1}}
\newcommand{\indep}{\perp\!\!\!\perp}

\newcommand{\code}[1]{\textsf{\small #1}}

\sectionfont{\noindent\normalfont\large\bf}
\subsectionfont{\noindent\normalfont\normalsize\bf}
\subsubsectionfont{\noindent\normalfont\it}

\begin{document}

\sgl

\pagestyle{empty}
\noindent {\Large \bf{Multinomial Inverse Regression for Text Analysis}}
\vskip 2cm

\noindent{\large Matt Taddy}\\
{\tt taddy@chicagobooth.edu}\\
{\it The  University of Chicago Booth School of Business}

\vskip 2cm

{\small 
\noindent{\sc Abstract:} 
Text data, including speeches, stories, and other document forms, are
often connected to {\it sentiment} variables that are of interest for
research in marketing, economics, and elsewhere.  It is also very high
dimensional and difficult to incorporate into statistical analyses.
This article introduces a straightforward framework of
sentiment-sufficient dimension reduction for text data.  Multinomial
inverse regression is introduced as a general tool for simplifying
predictor sets that can be represented as draws from a multinomial
distribution, and we show that logistic regression of phrase counts
onto document annotations can be used to obtain low dimension document
representations that are rich in sentiment information.  To facilitate
this modeling, a novel estimation technique is developed
for multinomial logistic regression with very high-dimension response.
In particular, independent Laplace priors with unknown variance are
assigned to each regression coefficient, and we detail an efficient
routine for maximization of the joint posterior over coefficients and
their prior scale.  This `gamma-lasso' scheme yields stable and effective estimation
for general high-dimension logistic regression, and we argue that it
will be superior to current methods in many settings.  Guidelines for
prior specification are provided, algorithm convergence is detailed,
and estimator properties are outlined from the perspective of the
literature on non-concave likelihood penalization.  Related work on
sentiment analysis from statistics, econometrics, and machine learning
is surveyed and connected.  Finally, the methods are applied in two
detailed examples and we provide out-of-sample prediction studies to
illustrate their effectiveness.  }

\vskip 2cm

\noindent {\footnotesize Taddy is an Associate Professor of
  Econometrics and Statistics and Neubauer Family Faculty Fellow at
  the University of Chicago Booth School of Business, and this work
  was partially supported by the IBM Corporation Faculty Research Fund
  at Chicago.  The author thanks Jesse Shapiro, Matthew Gentzkow,
  David Blei, Che-Lin Su, Christian Hansen, Robert Gramacy, Nicholas
  Polson, and anonymous reviewers for much helpful discussion.}

\newpage
\dbl

\pagestyle{plain}
\vskip 1cm
\section{Introduction}
\label{intro}

This article investigates the relationship between text data --
product reviews, political speech, financial news, or a personal blog
post -- and variables that are believed to influence its composition
-- product quality ratings, political affiliation, stock price, or
mood polarity.  Such language-motivating observable variables,
generically termed {\it sentiment} in the context of this article, are
often the main object of interest for text mining applications.  When,
as is typical, large amounts of text are available but only a small
subset of documents are annotated with known sentiment, this
relationship yields the powerful potential for text to act as a
stand-in for related quantities of primary interest.  On the other
hand, language data dimension (i.e., vocabulary size) is both very
large and tends to increase with the amount of observed text, making
the data difficult to incorporate into statistical analyses.  Our goal
is to introduce a straightforward framework of sentiment-preserving
dimension reduction for text data.

As detailed in Section \ref{background}.1, a common statistical
treatment of text views each document as an exchangeable collection of
phrase tokens.  In machine learning, these tokens are usually just
words (e.g., {\it tax}, {\it pizza}) obtained after stemming for
related roots (e.g., {\it taxation}, {\it taxing}, and {\it taxes} all
become {\it tax}), but richer tokenizations are also possible: for
example, we find it useful to track common {\it n}-gram word
combinations (e.g. bigrams {\it pay tax} or {\it cheese pizza} and
trigrams such as {\it too much tax}).  Under a given
tokenization each document is represented as $\bm{x}_i =
[x_{i1},\ldots,x_{ip}]'$, a sparse vector of counts for each of $p$
tokens in the vocabulary.  These token counts, and the associated
frequencies $\bm{f}_i = \bm{x}_i/m_i$ where $m_i = \sum_{j=1}^p
x_{ij}$, are then the basic data units for statistical text analysis.
In particular, the multinomial distribution for $\bm{x}_i$ implied by an
assumption of token-exchangeability can serve as the basis for
efficient dimension reduction.

Consider $n$ documents that are each annotated with a single sentiment
variable, $y_i$ (e.g., restaurant reviews accompanied by a one to five
star rating).  A naive approach to text-sentiment prediction would be
to fit a generic regression for $y_i |\bm{x}_i$.  However, given the
very high dimension of text-counts (with $p$ in the thousands or tens
of thousands), one cannot efficiently estimate this conditional
distribution without also taking steps to simplify $\bm{x}_i$.  We
propose an inverse regression (IR) approach, wherein the {\it inverse
  conditional distribution} for text given sentiment is used
to obtain low dimensional document scores that preserve information
relevant to $y_i$.

As an introductory example, consider the text-sentiment contingency table
built by collapsing token counts as $\bm{x}_y = \sum_{i:y_i=y}
\bm{x}_i$ for each $y \in \mc{Y}$, the
support of an ordered discrete sentiment variable.  A basic multinomial
inverse regression (MNIR) model is then
\begin{equation} \label{basic-mnir} \bm{x}_y \sim \mr{MN}(\bm{q}_y,
  m_y)~~\text{with}~~ q_{yj} = \frac{\exp[\alpha_j +
    y\varphi_j]}{\sum_{l=1}^p \exp[\alpha_l + y\varphi_l
    ]},~~\text{for}~~j=1,\ldots,p,~~y \in \mc{Y}
\end{equation}
where each $\mr{MN}$ is a $p$-dimensional multinomial distribution
with size $m_y = \sum_{i:y_i=y} m_i$ and probabilities $\bm{q}_y = [q_{y1},\ldots,q_{yp}]'$
that are a linear function of $y$ through a logistic link.  Under
conditions detailed in Section \ref{model}, the {\it
  sufficient reduction} (SR) score for $\bm{f}_i = \bm{x}_i/m_i$
is then
\begin{equation}\label{basic-sr}
z_i = \bs{\varphi}'\bm{f}_i ~\Rightarrow~ y_i\indep\bm{x}_i, m_i\mid z_i.
\end{equation}
Hence, given this SR projection, full $\bm{x}_i$ is ignored and
modeling the text-sentiment relationship becomes a univariate
regression problem.  This article's examples include linear,
$\ds{E}[y_i] = \beta_0 + \beta_1z_i$, quadratic, $\ds{E}[y_i] =
\beta_0 + \beta_1z_i + \beta_2z_i^2$, and logistic, $\mr{p}(y_i < a) =
\left(1 + \exp[\beta_0 + \beta_1z_i]\right)^{-1}$, forms for this {\it
  forward regression}, and SR scores should be straightforward to
incorporate into alternative regression models or structural equation
systems.  The procedure rests upon assumptions that allow for summary
tables wherein the text-sentiment relationship of interest can be
modeled as a logistic multinomial, but when such assumptions are
plausible, as we find common in text analysis, they introduce
information that should yield significant efficiency gains.

In estimating models of the type in (\ref{basic-mnir}), which 
involve many thousands of parameters, we propose use of fat-tailed and
sparsity-inducing independent Laplace priors for each coefficient 
$\varphi_j$.  To account for uncertainty about the appropriate level
of variable-specific regularization, each Laplace rate parameter
$\lambda_j$ is left unknown with a gamma hyperprior.  Thus,
for example,
\begin{equation}\label{basic-prior}
\pi(\varphi_j,\lambda_j) =
\frac{\lambda_{j}}{2}e^{-\lambda_{j}|\varphi_{j}|}
\frac{r^s}{\Gamma(s)}\lambda_j^{s-1}e^{-r\lambda_{j}},~~s,r,\lambda_j>0, 
\end{equation}
independent for each $j$ under a $\mr{Ga}(s,r)$ hyperprior
specification.  This departure from the usual shared-$\lambda$ model
is motivated in Section \ref{model}.3.

Fitting MNIR models is tough for reasons beyond the usual difficulties
of high dimension regression -- simply evaluating the large-response
likelihood is expensive due to the normalization in calculating each
$\bm{q}_i$.  As surveyed in Section \ref{estimation}, available
cross-validation (e.g., via solution paths) and fully Bayesian (i.e.,
through Monte-Carlo marginalization) methods for estimating
$\varphi_j$ under unknown $\lambda_j$ are prohibitively expensive.  A
novel algorithm is proposed for finding the joint posterior maximum
(MAP) estimate of both coefficients and their prior scale.  The
problem is reduced to log likelihood maximization for $\bs{\varphi}$
with a non-concave penalty, and it can be solved relatively
quickly through coordinate descent.  For example, given the prior in
(\ref{basic-prior}), the log likelihood implied by (\ref{basic-mnir})
is maximized subject to (i.e., minus) cost constraints
\begin{equation}\label{basic-cost}
c(\varphi_j) = s\log(1 + |\varphi_j|/r)
\end{equation}
for each coefficient.  This provides a powerful new estimation
framework, which we term the {\it gamma-lasso}.  The approach is very
computationally efficient, yielding robust SR scores in less than a
second for documents with thousands of unique tokens.
Indeed, although a full comparison is beyond the scope of this paper,
we find that the proposed algorithm can also be far superior to
current techniques for high-dimensional logistic regression in the
more common large-predictor (rather than large-response) setting.

This article thus includes two main methodological contributions. First,
Section \ref{model} introduces multinomial inverse regression as an IR
procedure for predictor sets that can be represented as draws from a
multinomial, and details its application to text-sentiment analysis.  This
includes full model specification and general sufficiency results, guidelines
on how text data should be handled to satisfy the MNIR model assumptions, and
our independent gamma-Laplace prior specification.  Second, Section
\ref{estimation} develops a novel approach to estimation in very high
dimensional logistic regression.  This includes details of coordinate descent
for joint MAP estimation of coefficients and their unknown variance,  and an
outline of estimator properties from the perspective of the literature on non-
concave likelihood penalization.  As background, Section \ref{background}
briefly surveys the literature on text mining and sentiment analysis, and on
dimension reduction and inverse regression.

The following section describes language pre-processing and
introduces two datasets that are used throughout to motivate and
illustrate our methods.  Performance comparison and detailed results
for these examples are then presented in Section \ref{results}.  Both
example datasets, along with all implemented methodology, are
available in the \code{textir} package for \code{R}.

\subsection{Data processing and examples}

Text is usually initially cleaned according to some standard
information retrieval criteria, and we refer the reader to
\citet{JuraMart2009} for an overview.  In this article, we simply
remove a limited set of stop words (e.g., {\it and} or {\it but}) and
punctuation, convert to lowercase, and strip suffixes from roots
according to the Porter stemmer \citep{Port1980}.  The main data
preparation step is then to parse clean text into informative language
tokens; as mentioned in the introduction, counts for these tokens are
the starting point for statistical analysis.  Most commonly
\citep[see, e.g.,][]{SrivSaha2009} the tokens are just words, such
that each document is treated as a vector of word-counts.  This is
referred to as the {\it bag-of-words} representation, since these
counts are summary statistics for language generated by exchangeable
draws from a multinomial `bag' of word options.

Despite its apparent limitations, the token-count framework can be
made quite flexible through more sophisticated tokenization.  For
example, in the {\it$N$-gram} language model words are drawn from a Markov
chain of order $N$ \citep[see, e.g.,][]{JuraMart2009}.  A document is
then summarized by its length-$N$ word sequences, or $N$-gram tokens,
as these are sufficient for the underlying Markov transition
probabilities.  Our general practice is to count common unigram,
bigram, and trigram tokens (i.e., words and 2-3 word phrases).
Another powerful technique is to use domain-specific knowledge to
parse for phrases that are meaningful in the context of a specific
field.  \citet{TallOKan2012} present one such approach for
tokenization of legal agreements; for example, they use any
conjugation of the word {\it act} in proximity of {\it God} to
identify a common {\it Act of God} class of carve-out provisions.
Finally, work such as that of \citet{PoonDomi2009} seeks to parse
language according to semantic equivalence.

Thus while we focus on token-count data, different language models are
able to influence analysis through tokenization rules.  And although
separation of parsing from statistical modeling limits our ability to
quantify uncertainty, it has the appealing effect of allowing text
data from various sources and formats to all be analyzed within a
multinomial likelihood framework.

\subsubsection{Ideology in political speeches}

This example originally appears in \citet[GS;][]{GentShap2010} and
considers text of the 109$^\text{th}$ (2005-2006) Congressional
Record.  For each of the 529 members of the United States House and
Senate, GS record usage of phrases in a list of 1000 bigrams and
trigrams.  Each document corresponds to a single person.  The
sentiment of interest is political partisanship, where party
affiliation (Republican, Democrat, or Independent) provides a simple
indicator and a higher-fidelity measure is calculated as the two-party
vote-share from each speaker's constituency (congressional district
for representatives; state for senators) obtained by George W. Bush in
the 2004 presidential election. Note that token vocabulary in this
example is influenced by sentiment: GS built contingency tables for
bigram and trigram usage by party, and kept the top 1000 `most
partisan' phrases according to ranking of their Pearson $\chi^2$-test
statistic.

Define phrase frequency {\it lift} for a given group as
$\bar{f}_{j\mr{G}}/ \bar{f}_{j}$, where $\bar{f}_{j\mr{G}}$ is mean
frequency for phrase $j$ in group $\mr{G}$ and $\bar{f}_{j} =
\sum_{i=1}^n f_{ij}/n$ is the average across all documents.  The
following tables show top-five lift phrases used at least once by each
party.

\begin{center}\sgl\footnotesize
\begin{minipage}{2.5in}
\hfill {\sc \normalsize Democratic Frequency Lift}
\vskip .2cm
\hfill \begin{tabular}{r | l }                                   
{\sf  congressional.hispanic.caucu }&  2.163\\
{\sf  medicaid.cut  }& 2.154 \\  
{\sf  clean.drinking.water }& 2.154 \\ 
{\sf  earth.day }& 2.152\\   
{\sf tax.cut.benefit  }& 2.149\\           
\end{tabular}  
\end{minipage}
\hskip .5in
\begin{minipage}{2.5in}
\hfill {\sc\normalsize Republican Frequency Lift}
\vskip .2cm
\hfill \begin{tabular}{r | l }                                   
{\sf  ayman.al.zawahiri }&  1.850 \\
{\sf  america.blood.cent }& 1.849 \\  
{\sf  million.budget.request }& 1.847 \\ 
{\sf  million.illegal.alien }& 1.846 \\   
{\sf  temporary.worker.program }& 1.845\\           
\end{tabular}  
\end{minipage}
\end{center}

\subsubsection{On-line restaurant reviews}

This dataset, which originally appears in the topic analysis of
\citet{MauaCozm2009}, contains 6260 user-submitted restaurant reviews
(90 word average) from  {\tt www.we8there.com}.  The reviews
are accompanied by a five-star rating on four specific aspects of
quality -- {\it food}, {\it service}, {\it value}, and {\it
  atmosphere} -- as well as the {\it overall experience}.  After
tokenizing the text into bigrams (based on a belief that modifiers
such as {\it very} or {\it small} would be useful here), we discard
phrases that appear in less than ten reviews and documents which do
not use any of the remaining phrases.  This leaves a dataset of 6147
review counts for a token vocabulary of 2978 bigrams.  Top-five lift
phrases that occur at least once in both positive ({\it overall experience}
$>$ 3) and negative ({\it overall experience} $<$ 3) reviews are
below.

\begin{center}
\sgl\footnotesize
\begin{minipage}{2.5in}
\hfill {\sc \normalsize Negative Frequency Lift}
\vskip .2cm
\hfill \begin{tabular}{r | l }                                   
{\sf  food poison } &  5.402\\
{\sf  food terribl  } & 5.354 \\  
{\sf  one worst } & 5.339 \\ 
{\sf  spoke manag } & 5.318\\   
{\sf after left  } & 5.285\\           
\end{tabular}  
\end{minipage}
\hskip .5in
\begin{minipage}{2.5in}
\hfill {\sc\normalsize Positive Frequency Lift}
\vskip .2cm
\hfill \begin{tabular}{r | l }                                   
{\sf  worth trip  } &  1.393 \\
{\sf  everi week } & 1.390 \\  
{\sf  melt mouth } & 1.389 \\ 
{\sf  alway go } & 1.389 \\   
{\sf  onc week } & 1.389\\           
\end{tabular}  
\end{minipage}
\end{center}

\section{Background}
\label{background}

This section briefly reviews the
relevant literatures on sentiment analysis and inverse
regression. Additional background is in the appendices
and material specific to estimation is in Section \ref{estimation}.

\subsection{Analysis of sentiment in text}

As already outlined, we use {\it sentiment} to refer to any
variables related to document composition.  Although broader
than its common `opinion polarity' usage, this definition as `sensible
quality' fits our need to refer to the variety of quantities that may be
correlated with text.

Much of existing work on sentiment analysis uses word frequencies as
predictors in generic regression and classification algorithms,
including support vector machines, principle components (PC)
regression, neural networks, and penalized least-squares.  Examples
from this machine learning literature can be found in the survey by
\citet{PangLee2008} and in the collection from \citet{SrivSaha2009}.
In the social sciences, research on ideology in political text
includes both generic classifiers \citep[e.g.,][]{YuKaufDier2008} and
analysis of contingency tables for individual terms
\citep[e.g.,][]{LaveBenoGarr2003} \citep[machine learning researchers,
such as][have also made contributions in this area]{ThomPangLee2006}.
In economics, particularly finance, it is more common to rely on
weighted counts for pre-defined lists of terms with positive or
negative {\it tone}; examples of this approach include
\citet{Tetl2007} and \citet{LougMcDo2011} \citep[again, machine
learners such as][have also studied prediction for
finance]{BollMaoZeng2010}.

These approaches all have drawbacks: generic regression does nothing
to leverage the particulars of text data, independent analysis of many
contingency tables leads to multiple-testing issues, and pre-defined
word lists are subjective and unreliable.  A more promising strategy
is to use text-specific dimension reduction based upon the multinomial
implied by exchangeability of token-counts.  For example, a {\it topic
  model} treats documents as drawn from a multinomial distribution
with probabilities arising as a weighted combination of `topic'
factors.  Thus $ \bm{x}_i \sim \mr{MN}(\omega_{i1} \bs{\theta}_{1} +
\ldots +\omega_{iK} \bs{\theta}_{K}, m_i)$, where topics
$\bs{\theta}_k = [\theta_{k1} \cdots \theta_{kp}]'$ and weights
$\bs{\omega}_i$ are probability vectors. This framework, also known as
{\it latent Dirichlet allocation} (LDA), has been widely used in text
analysis since its introduction by \citet{BleiNgJord2003}.

The low dimensional topic-weight representation (i.e.,
$\bs{\omega}_i$) serves as a basis for sentiment analysis
in the original Blei et al.\!  article, and has been used in this way
by many since.  The approach is especially popular in political
science, where work such as that of \citet{Grim2010} and
\citet{QuinMonrColaCresRade2010} investigates political interpretation
of latent topics (these authors restrict $\omega_{ik} \in \{0,1\}$
such that each document is drawn from a single topic).  Recently,
\citet{BleiMcAu2010} have introduced supervised LDA (sLDA) for joint
modeling of text and sentiment.  In particular, they augment topic
model with a forward regression $y_i = f(\bs{\omega}_i)$, such that
token counts and sentiment are connected through shared topic-weight
factors.

Finally, our investigation was originally motivated by a desire to
build a model-based version of the specific {\it slant} indices
proposed by \citet{GentShap2010}, which are part of a general
political science literature on quantifying partisanship through
weighted-term indices \cite[e.g.,][]{LaveBenoGarr2003}.  Appendix A.1
shows that the GS indices can be written as summation of phrase
frequencies loaded by their correlation with measured partisanship
(e.g., vote-share), such that slant is equivalent to first-order
partial least-squares \citep[PLS;][]{Wold1975}.

\subsection{Inverse regression and sufficient reduction}

This article is based on a notion that, given the high dimension of
text data, it is not possible to efficiently estimate conditional
response $y|\bm{x}$ without finding a way to simplify $\bm{x}$.  The
same idea motivates many of the techniques surveyed above, including
LDA and sLDA, PLS/slant, and PC regression.  A framework to unify
techniques for dimension reduction in regression can be found in
Cook's 2007 \nocite{Cook2007} overview of {\it inverse regression},
wherein inference about the multivariate conditional distribution
$\bm{x}|y$ is used to build low dimension summaries for $\bm{x}$. 

Suppose that $\bm{v}_i$ is a $K$-vector of {\it response
  factors} through which $\bm{x}_i$ depends on $y_i$ (i.e., $\bm{v}_i$
is a possibly random function of $y_i$).  Then Cook's linear IR
formulation has $\bm{x}_i = \bs{\Phi} \bm{v}_i + \bs{\epsilon}_i$,
where $\bs{\Phi} = [\bs{\varphi}_1 \cdots \bs{\varphi}_K]$ is a
$p\times K$ matrix of inverse regression coefficients and
$\bs{\epsilon}_i$ is $p$-vector of error terms. Under certain
conditions on $\mr{var}(\bs{\epsilon}_i)$, detailed by Cook, the
projection $\bm{z}_i = \bs{\Phi}'\bm{x}_i$ provides a {\it sufficient
  reduction} (SR) such that $y_i$ is independent of $\bm{x}_i$ given
$\bm{z}_i$.  As this implies $\mr{p}(\bm{x}_i | \bs{\Phi}'\bm{x}_i,
y_i) = \mr{p}(\bm{x}_i | \bs{\Phi}'\bm{x}_i)$, SR corresponds to the
classical definition of sufficiency for `data' $\bm{x}_i$ and
`parameter' $y_i$, but is conditional on unknown $\bs{\Phi}$ that must
be estimated in practice.  When such estimation is feasible, the
reduction of dimension from $p$ to $K$ should make these SR {\it
  projections} easier to work with than the original predictors.

Many approaches to dimension reduction can be understood in context of
this linear IR model: PC directions arise as SR projections for the
maximum likelihood solution when $\bm{v}_i$ is unspecified \citep[see,
e.g.,][]{Cook2007} and, following our discussion in A.1, the first PLS
direction is the SR projection for least-squares fit when $\bm{v}_i =
y_i$.  A closely related framework is that of factor analysis, wherein
one seeks to estimate $\bm{v}_i$ directly rather than project
$\bm{x}_i$ into its lower dimensional space.  By augmenting estimation
with a forward model for $y_i | \bm{v}_i$ researchers are able to
build {\it supervised factor models}; see, e.g., \citet{West2003}.

The innovation of inverse regression, from Cook's 2007 paper and in
earlier work including \citet{Li1991} and \citet{BuraCook2001}, is to
investigate the SR projections that result from explicit specification
for $\bm{v}_i$ as a function of $y_i$.  Cook's {\it principle fitted
  components} are derived for a variety of functional expansions of
$y_i$, \citet{LiCookTsai2007} interprets PLS within an IR framework,
and the {\it sliced inverse regression} of \citet{Li1991} defines
$\bm{v}_i$ as a step-function expansion of $y_i$.  Since in each case
the $\bm{v}_i$ are conditioned upon, these IR models are more
restrictive than the random joint forward-inverse specification of
supervised factor models.  But if the IR model assumptions are
satisfied then its parsimony should lead to more efficient inference.

Instead of a linear equation, dimension reduction for text data is
based on multinomial models.  Following the topic model factor
specification, LDA is akin to PC analysis for multinomials and sLDA is
the corresponding supervised factor model.  However, existing work on
non-Gaussian inverse regression relies on conditional independence;
for example, \citet{CookLi2009} use single-parameter exponential
families to model each $x_{ij}|\bm{v}_i$.  To our knowledge, no-one
has investigated SR projections based on the multinomial predictor
distributions that arise naturally for text data.  Hence, we seek to
build a multinomial inverse regression framework.

\section{Modeling}
\label{model}

The subject-specific multinomial inverse regression model has, for
$i=1,\ldots,n$:
\begin{equation}\label{full-mnir}
  \bm{x}_i \sim \mr{MN}(\bm{q}_i, m_i)~~\text{with}~~
  q_{ij} = \frac{e^{\eta_{ij}}}{\sum_{l=1}^p
    e^{\eta_{il}}},~~j=1,\ldots,p,~~\text{where}~~\eta_{ij} = \alpha_j +
  u_{ij} +
  \bm{v}_i'\bs{\varphi}_j.
\end{equation}
This generalizes (\ref{basic-mnir}) with the introduction of
$K$-dimensional response factors $\bm{v}_i$ and subject effects
$\bm{u}_i = [u_{i1} \cdots u_{ip}]'$.  Section \ref{model}.1 derives
sufficient reduction results for projections $ \bm{z}_im_i =
\bs{\Phi}'\bm{x}_i$, where $\bs{\Phi}' = [\bs{\varphi}_1,\cdots
\bs{\varphi}_p]$.  Section \ref{model}.2 then describes application of
these results in text analysis and outlines situations where (\ref{full-mnir}) can be
replaced with a collapsed model as in (\ref{basic-mnir}).  Finally,
\ref{model}.3 presents prior specification for these very high
dimensional regressions.

\subsection{Sufficient reduction in multinomial inverse regression}

This section establishes classical sufficiency-for-$y$ (conditional on
IR parameters) for projections derived from the 
model in (\ref{full-mnir}).  The main result follows, due to use of a
logit link on $\bs{\eta}_i = [\eta_{i1}\cdots\eta_{ip}]'$, from 
factorization of the multinomial's natural exponential
family parametrization.
\begin{prop}\label{SRprop}
  Under the model in (\ref{full-mnir}), conditional on $m_i$ and
  $\bm{u}_i$ \vspace{-.5cm}\[ y_i \indep \bm{x}_i \mid \bm{v}_i
  \Rightarrow y_i \indep \bm{x}_i \mid \bs{\Phi}'\bm{x}_i.
\]
\end{prop}
\begin{proof}\vspace{-.25cm} 
  Setting $\alpha_{ij} = \alpha_j + u_{ij}$ and suppressing $i$, the
  likelihood is ${m \choose \bm{x}} \exp\left[ \bm{x}'\bs{\eta} -
    A(\bs{\eta})\right] = {m \choose \bm{x}} e^{\bm{x}'\bs{\alpha} }
  \exp\left[(\bm{x} '\bs{\Phi})\bm{v} - A(\bs{\eta})\right] =
  h(\bm{x}) g(\bs{\Phi}'\bm{x}, \bm{v})$, where $A(\bs{\eta}) =
  m\log\left[\sum_{j=1}^pe^{\eta_j}\right]$.  Hence, the usual
  sufficiency factorization \citep[e.g.,][2.21]{Sche1995} implies
  $\mr{p}(\bm{x} | \bs{\Phi}'\bm{x} , \bm{v}) = \mr{p}(\bm{x} |
  \bs{\Phi}'\bm{x} )$, and $\bm{v}$ is independent of $\bm{x}$ given
  $\bs{\Phi}'\bm{x}$.  Finally, $\mr{p}(y | \bm{x}, \bs{\Phi}'\bm{x})
  = \int_{\bm{v}} \!\mr{p}(y|\bm{v}) d\mr{P}(\bm{v}| \bs{\Phi}'\bm{x})
  = \mr{p}(y |\bs{\Phi}'\bm{x})$.
\end{proof}

\vspace{-.5cm} 
Second, it is standard in text analysis to control for
document size by regressing $y_i$ onto frequencies rather than
counts.  Fortunately, our sufficient reductions survive this
transformation.
\begin{prop}\label{SRf}
  If $ y_i \indep \bm{x}_i \mid \bs{\Phi}'\bm{x}_i, m_i$ and $\mr{p}(y
  \mid \bm{x}_i) = \mr{p}(y_i \mid \bm{f}_i)$, then $y_i \indep
  \bm{x}_i \mid \bm{z}_i = \bs{\Phi}'\bm{f}_i$.
\end{prop}
\begin{proof}
  We have that each of $\bm{f}$ and $[\bs{\Phi}'\bm{f}, m]$ are
  sufficient for $y$ in $\mr{p}(\bm{x}| y) =
  \mr{MN}(\bm{q},m)\mr{p}(m|y)$.  Under conditions of
  \citet[][6.3]{LehmShef1950}, there exists a minimal sufficient
  statistic $T(\bm{x})$ and functions $g$ and $\tilde g$ such that
  $g(\bm{f}) = T(\bm{x}) = \tilde g(\bs{\Phi}'\bm{f}, m)$.
  Having $\tilde g$ vary with $m$, while $g(\bm{f})$ does not, implies
  that the map $\bs{\Phi}'\bm{f}$ has introduced such dependence.  But
  since $m$ cannot be recovered from $\bm{f}$, this must be false.
  Thus $\tilde g = \tilde g(\bs{\Phi}'\bm{f})$,
  and $\bm{z} = \bs{\Phi}'\bm{f}$ is sufficient for $y$.
\end{proof}

\subsection{MNIR for sentiment in text: collapsibility and random effects}

For text-sentiment response factor specification, we focus on
untransformed $v_i=y_i$ and discretized $v_i = \mr{step}(y_i)$ along
with their analogues for multivariate sentiment.  The former is
appropriate for categorical sentiment (e.g., political party, or 1-5
star rating) and, for reasons discussed below, the latter is used with
continuous sentiment (e.g., vote-share is rounded to the nearest whole
percentage, and in general one can bin and average $y$ by quantiles).
Regardless, our methods apply under generic $\bm{v}(y_i)$ including,
e.g., the expansions of \citet{Cook2007}.

Given this setting of discrete $\bm{v}_i$, MNIR estimation can often
be based on the {\it collapsed} counts that arise by aggregating
within factor level combinations.  For example, since sums of
multinomials with equal probabilities are also multinomial, given
shared intercepts (i.e., $u_{ij}=0$) and writing the support of
$\bm{v}_i$ as $\mc{V}$, the likelihood for the model in
(\ref{full-mnir}) is exactly the same as that from, for $\bm{v} \in
\mc{V}$ with $\bm{x}_\bm{v} = \sum_{i:\bm{v}_i=\bm{v}}\bm{x}_i$ and
$m_\bm{v} = \sum_{i:\bm{v}_i=\bm{v}}m_i$,
\begin{equation}\label{vbin}
\bm{x}_\bm{v} \sim
\mr{MN}(\bm{q}_\bm{v}, m_\bm{v}),~~\text{where}~
  q_{\bm{v}j} = \frac{e^{\eta_{\bm{v}j}}}{\sum_{l=1}^p
    e^{\eta_{\bm{v}l}}} ~~\text{and}~~\eta_{\bm{v}j} = \alpha_j +
\bm{v}\bs{\varphi}_j.
\end{equation}
Since pooling documents in this way leaves only as many `observations'
as there are levels in the support of $\bm{v}_i$, it can lead to
dramatically less expensive estimation.

Under the marginal model of (\ref{vbin}), $\bs{\Phi}$ is the {\it population average}
effect of $\bm{v}$ on $\bm{x}$.  One needs to be careful in when and
how estimates from this model are used in SR projection, since
conditional document-level validity of these results is subject to the
usual collapsibility requirements for analysis of categorical data
\citep[e.g.,][]{BishFienHoll1975}.  In particular, omitted variables
must be conditionally independent of $\bm{x}_i$ given $\bm{v}_i$; this
can usually be assumed for sentiment-related variables (e.g., a
congress person's voting record is ignored given their
party and vote-share).  Covariates that act on $\bm{x}_i$ independent
of $\bm{v}_i$ should be included in MNIR, as part of the equation for
subject effects $\bm{u}_{i}$ (e.g., although it is not considered in
this article, it might be best to condition on geography when
regressing political speech onto partisanship).  The sufficient
reduction result of (\ref{SRprop}) is then conditional on these
sentiment-independent variables, such that they (or their SR
projection) {\it may} need to be used as inputs in forward regression.

It is often unreasonable to assume that known factors account for all
variation across documents, and treating the $\bm{u}_i$ of
(\ref{full-mnir}) as random effects independent of $\bm{v}_i$ provides
a mechanism for explaining such heterogeneity and understanding its
effect on estimation.  Omitting $\bm{u}_i \indep \bm{v}_i$ tends to
yield estimated $\bs{\Phi}$ that is attenuated from its correct
document-specific analogue \citep{GailWieaPian1984}, although the
population-average estimators can be reliable in some settings; for
example, \citet{ZegeLianSelf1985} show consistency for the stationary
distribution effect of covariates when the $\bm{u}_i$ encode temporal
dependence (such as that between consecutive tokens in an $N$-gram
text model).  When their influence is considered negligible, it is
common to simply ignore the random effects in estimation.  In this
article we also consider modeling $e^{u_{ij}}$ as independent gamma
random variables, and use this to motivate a prior in \ref{model}.3
for the marginal random effects in a collapsed table.  Another option
would be to incorporate latent topics into MNIR and parametrize
$\bm{u}_{i}$ through a linear factor model; this is especially
appealing since SR projections onto estimated factor scores could then
be used in forward regression.

This last point -- on random effects and forward regression -- is
important: when $\bs{\Phi}$ is estimated with random effects, Section
\ref{model}.1 only establishes sufficiency of $\bm{z}_i$ conditional
on $\bm{u}_i$.  Marginal sufficiency would follow from
$\mr{p}(\bm{v}_i | \bm{u}_i, \bs{\Phi}'\bm{x}_i) = \mr{p}(\bm{v}_i |
\bs{\Phi}'\bm{x}_i)$, which for $\bm{u}_i \indep \bm{v}_i$ requires
$\bm{u}_i \indep \bs{\Phi}'\bm{x}_i$.  Thus, information about
$\bm{v}_i$ from this marginal dependence is lost when (as is usually
necessary) $\bm{u}_i$ is omitted in regression of $\bm{v}_i$ onto
$\bm{z}_i$.  Section \ref{results} shows that random effects in MNIR
can be beneficial even if they are then ignored in forward regression.
However, SR projection onto parametric representations of $\bm{u}_i$
is an open research interest.

It is clear that there are many relevant issues to consider when
assessing an MNIR model, and it is helpful to have our sentiment
regression problem placed within the well studied framework of
contingency table analysis \citep[e.g.,][is a general
reference]{Agre2002}.  Ongoing work centers on inference according to
specific dependence structures or random effect parametrizations.
However, as illustrated in Section \ref{results}, even very simple
MNIR models -- measuring population average effects -- allow SR
projections that are powerful tools for forward prediction.

\subsection{Prior specification}

To complete the MNIR model, we provide prior distributions for the
intercepts $\bs{\alpha}$, loadings $\bs{\Phi}$, and possible random
effects $\bm{U} = [\bm{u}_{\bm{v}_1} \cdots \bm{u}_{\bm{v}_d}]'$,
where $d$ is the number of points in $\mc{V}$. 

First, each phrase intercept is assigned an independent standard
normal prior, $\alpha_j \sim \mr{N}(0,1)$.  This serves to identify
the logistic multinomial model, such that there is no need to specify
a {\it null} category, and we have found it diffuse enough to
accomodate category frequencies in a variety of text and non-text
examples.  Second, we propose independent Laplace priors for each
$\varphi_{jk}$, with coefficient-specific precision (or `penalty')
parameters $\lambda_{jk}$, such that $\pi(\varphi_{jk}) =
\lambda_{jk}/2\exp(-\lambda_{jk}|\varphi_{jk}|) $ for $j=1\ldots p$
and $k=1\ldots K$.  The implied prior standard deviation for
$\varphi_{jk}$ is $\sqrt{2}/\lambda_{jk}$.  Each $\lambda_{jk}$ is
then assigned a conjugate gamma hyperprior $\mr{Ga}(\lambda_{jk}; s,
r) = r^s/\Gamma(s) \lambda_{jk}^{s-1}e^{-r\lambda_{jk}}$, yielding the
joint gamma-Laplace prior introduced in (\ref{basic-prior}).
Hyperprior shape, $s$, and rate, $r$, imply expectation $s/r$ and
variance $s/r^2$ for each $\lambda_{jk}$.

As an example specification, consider variation in empirical token
probabilities by level of the logical variables `\code{party
  $=$ republican}' for congressional speech and `\code{rating $>$ 3}'
for we8there reviews. Standard deviation of finite $\log(\hat
q_{\tt true,j}/\hat q_{\tt false, j})$ across tokens is 1.9 and 1.4
respectively, and given variables normalized to have $\mr{var}(v)=1$
these deviations in log-odds correspond to a jump of two in $v$ (from
approximately -1 to 1).  Hence, a coefficient standard deviation of
around 0.7, implying $\ds{E}[\lambda_{jk}]= 2$, is at the conservative
(heavy penalization) end of the range indicated by informal data
exploration, recommending the exponential $\mr{Ga}(1,1/2)$ as a
penalty prior specification.  In Section \ref{results} we also consider shapes
of 1/10 and 1/100, thus decreasing $\ds{E}[\lambda_{jk}]$ by two
orders of magnitude, and find performance robust to these
changes.

The above models have, with $s \leq 1$, hyperprior densities for
$\varphi_{jk}$ that are increasing as the penalty approaches zero
(i.e., at MLE estimation).  This strategy has performed well in many
applications, both for text analysis and otherwise, when dimension is
not much larger than $10^3$.  However, in examples with vocabulary
sizes reaching $10^5$ and higher, it is useful to increase both shape
and rate for fast convergence and to keep the number of non-zero term
loadings manageably small.  As an informal practical recipe, if estimated
$\bs{\Phi}$ is less sparse than desired and you suspect overfit,
increase $s$.  Following the discussion in \ref{estimation}.3 on
hyperprior variance and algorithm convergence, if the optimization is
taking too long or getting stuck in a minor mode, multiply both $s$
and $r$ by a constant to keep $\ds{E}[\lambda_{jk}]$ unchanged
while decreasing $\mr{var}[\lambda_{jk}]$.

Finally, we use $\exp[u_{ij}] \sim \mr{Ga}(1,1)$ independent for each
$i$ and $j$ as an illustrative random effect model.  Considering
$e^{u_{ij}}$ as a multiplier on relative odds, its mode at zero
assumes some tokens are inappropriate for a given document, the mean
of one centers the model on a shared intercept, and the fat right tail
allows for occasional large counts of otherwise rare tokens.  Counts
are not immediately collapsable in the presence of random effects, but
assumptions on the generating process for $\bm{x}_i$ unconditional on
$m_i$ can be used to build a prior model for their effect on
aggregated counts: if each $x_{ij}$ is drawn independent from a
Poisson $\mr{Po}(e^{\alpha_j + u_{ij} + \bm{v}_i\varphi_j})$ with
$\exp[u_{ij}]\sim \mr{Ga}(1,1)$, and $n_\bm{v} = \sum_i
\ds{1}_{[\bm{v}_i = \bm{v}]}$, then $x_{\bm{v}j} \sim
\mr{Po}(e^{\alpha_j + u_{\bm{v}j} +\bm{v}\varphi_j})$ with
$\exp[u_{\bm{v}j}]\stackrel{ind}{\sim} \mr{Ga}(n_\bm{v},1)$.  For
convenience, we use a log-Normal approximation to the gamma and
specify $u_{\bm{v},j} \sim \mr{N}(\log(n_{\bm{v}}) -
0.5\sigma_{\bm{v}}^2, \sigma_{\bm{v}}^2)$ with $\sigma_{\bm{v}}^2 =
\log(n_{\bm{v}}+1) - \log(n_{\bm{v}})$.  Note that $\sigma_{\bm{v}}^2
\rightarrow 0$ as $n_\bm{v}$ grows, leading to static $u_{\bm{v},j}$
whose effect is equivalent to multiplying both numerator and
denominator of $\exp[\eta_{\bm{v},j}]/\sum_{l}\exp[\eta_{\bm{v},l}]$
by a constant.  Thus modeling random effects is unnecessary {\it under
  our assumed model} after aggregating large numbers of observations.

\subsubsection{Motivation for independent gamma-Laplace priors}

One unique aspect of this article's approach is the use of independent
gamma-Laplace priors for each regression coefficient $\varphi_{jk}$.
Part of the specification should not be surprising: the Laplace 
provides, as a scale-mixture of normal densities, a widely used robust
alternative to the conjugate normal prior
\citep[e.g.,][]{CarlPolsStof1992}.  It also encourages sparsity in
$\bs{\Phi}$ through a sharp density spike at $\varphi_{jk} = 0$, and
MAP inference with fixed $\lambda_{jk}$ is equivalent to likelihood
maximization under an $L_1$-penalty in the {\it lasso} estimation
and selection procedure of \citet{Tibs1996}.  Similarly, conjugate
gamma hyperpriors are a common choice in Bayesian inference for lasso
regression \citep[e.g.,][]{ParkCase2008}.

However, our use of independent precision for each coefficient, rather
than a single shared $\lambda$, is a departure from standard practice.
We feel that this provides a better representation of prior
utility, and it avoids the overpenalization that can occur when
inferring a single coefficient precision on data with a large
proportion of spurious regressors.  In their recent work on
the Horseshoe prior, \citet{CarvPolsScot2010} illustrate general
practical and theoretical advantages of an independent parameter
variance specification.  As detailed in Section \ref{estimation}, our
model also yields an estimation procedure, labeled the
{\it gamma-lasso}, that corresponds to likelihood maximization under a
specific nonconcave penalty; the estimators thus inherit properties
deemed desirable by authors in that literature \citep[beginning
from][]{FanLi2001}.

Finally, given the common reliance on cross-validation (CV) for lasso
penalty selection, it is worth discussing why we choose to do
otherwise.  First, our independent $\lambda_{jk}$ penalties would
require a CV search of impossibly massive dimension.  Moreover, CV is
just an estimation technique and, like any other, is sensitive to the
data sample on which it is applied.  As an illustration, Section
\ref{results}.1 contains an example of CV-selected penalty performing
far worse in out-of-sample prediction than those inferred under a wide
range of gamma hyperpriors.  CV is also not scaleable: repeated
training and validation is infeasible on truly large applications
(i.e., when estimating the model once is expensive).  That said, one
may wish to use CV to choose $s$ or $r$ in the hyperprior; since
results are less sensitive to these parameters than they are to a
fixed penalty, a small grid of search locations should suffice.

\section{Estimation}
\label{estimation}

Following our model specification in Section \ref{model}, the full
posterior distribution of interest is
\begin{equation}
  \mr{p}(\bs{\Phi},\bs{\alpha},\bs{\lambda}, \bm{U} \mid \bm{X}, \bm{V}) \propto
  \prod_{i=1}^n \prod_{j=0}^p 
    q_{ij}^{x_{ij}} \pi(u_{ij})  \mr{N}(\alpha_{j} ; 0, \sigma^2_\alpha)
    \prod_{k=1}^K\mr{GL}(\varphi_{jk}, \lambda_{jk})
\label{post}
\end{equation}
where $q_{ij} = \exp[\eta_{ij}]/\sum_{l=1}^p \exp[\eta_{il}]$ with
$\eta_{ij} = \alpha_j + u_{ij} + \sum_{k=1}^Kv_{ik}\varphi_{jk}$ and
$\mr{GL}$ is our gamma-Laplace joint coefficient-penalty prior
$\mr{Laplace}(\varphi_{jk} ; \lambda_{jk})\mr{Ga}(\lambda_{jk}; r,
s)$.  We only consider here $u_{ij}=0$ or $u_{ij}
\stackrel{ind}{\sim}\mr{N}(0,\sigma^2_{i})$ for $\pi(u_{ij})$,
although sentiment-independent covariates can also be included trivially as
additional dimensions of $\bm{v}_i$.  Note that `$i$' denotes an
observation, but that in MNIR this will often be a combination of
documents after the aggregation of Section \ref{model}.2.

Bayesian analysis of logistic regression typically involves posterior
simulation, e.g. through Gibbs sampling with latent variables
\citep{HolmHeld2006} or Metropolis sampling with
posterior-approximating proposals \citep{RossAlleMcCu2005}.  Despite
recent work on larger datasets and sparse signals
\citep[e.g.,][]{GramPols2010}, our experience is that these methods
are too slow for text analysis applications.  Even the more modest 
goal of posterior maximization presents considerable difficulty:
unlike the usual high-dimension logistic regression examples, where
$K$ is big and $p$ is small, our large response leads to a likelihood
that is expensive to evaluate (due to normalization of each
$\bm{q}_i$) and has a dense information matrix (from
\ref{estimation}.2, $\partial^2\log\mr{LHD} /\partial\varphi_{jk} =
\sum_{i=1}^n m_iv_{ik}^2q_{ij}(1-q_{ij})$, which will not be zero
unless $v_{ik}$ is).  As a result, commonly used path algorithms that
solve over a grid of shared $\lambda$ values \citep[e.g.,][as
implemented in \code{glmnet} for \code{R}]{FrieHastTibs2010} do not
work even for the small examples of this article.

We are thus motivated to develop super efficient estimation for sparse
logistic regression.  The independent gamma-Laplace priors of Section
\ref{model}.3 are the first crucial aspect of our approach: it remains
necessary to choose hyperprior $s$ and $r$, but results are robust enough to
misspecification that basic defaults can be applied. Section
\ref{estimation}.1 derives the gamma-lasso (GL) non-concave penalty that
results from MAP estimation under this prior.  Second, Section
\ref{estimation}.2 describes a coordinate descent algorithm for fast negative
log posterior minimization wherein the GL penalties are incorporated at no
extra cost over standard lasso regression. Lastly, \ref{estimation}.3
considers conditions for posterior log concavity and convergence.

\subsection{Gamma-lasso penalized regression}

Our estimation framework relies upon recognition that optimal
$\lambda_{jk}$ can always be written as a function of $\varphi_{jk}$,
and thus does not need to be explicitly solved for in the joint
objective.
\begin{prop}\label{penprop}
  MAP estimation for $\bs{\Phi}$ and $\bs{\lambda}$ under the
  independent gamma-Laplace prior model in (\ref{post}) is equivalent
  to minimization of the negative log likelihood for $\bs{\Phi}$ subject
  to costs
\begin{equation}\label{cost}
c(\bs{\Phi}) = \sum_{j=1}^p \sum_{k=1}^K c(\varphi_{jk}),~~\text{where}~~
c(\varphi_{jk}) = s\log( 1 + |\varphi_{jk}|/r )
\end{equation}
\end{prop}
\begin{proof}
  Under conjugate gamma priors, the conditional posterior mode for
  each $\lambda_{jk}$ given $\varphi_{jk}$ is available as
  $\lambda(\varphi_{jk}) = s/(r + |\varphi_{jk}|)$.  Any joint
  maximizing solution $[\bs{\hat\Phi},\bs{\hat\lambda}]$ for
  (\ref{post}) will thus consist of $\hat{\lambda}_{jk} =
  \lambda(\hat\varphi_{jk})$; otherwise, it is always possible to increase
  the posterior by replacing $\hat\lambda_{jk}$.  Taking the negative
  log of (\ref{basic-prior}) and removing constant terms, the
  influence of a $\mr{GL}(\lambda_{jk},\varphi_{jk})$ prior on the
  negative log posterior is $- s \log(\lambda_{jk}) + (r +
  |\varphi_{jk}|)\lambda_{jk}$, which becomes $-s\log\left[ (s/r) / ( 1
    + |\varphi_{jk}|/r ) \right] + s \propto s\log( 1 +
  |\varphi_{jk}|/r )$ after replacing $\lambda_{jk}$ with
  $\lambda(\varphi_{jk})$.
\end{proof}

The implied penalty function is drawn in the left panel of Figure
\ref{soln}.  Given its shape -- everywhere concave with a sharp spike
at zero -- our gamma-lasso estimation fits within the general
framework of nonconcave penalized likelihood maximization as outlined
in \citet{FanLi2001} and studied in many papers since.  In particular,
$c(\varphi_{jk})$ can be seen as a reparametrization of the
`log-penalty' described in \citet[][eq. 10]{MazuFrieHast2011}, which
is itself introduced in \citet{Frie2008} as a generalization of the
elastic net.  Viewing estimation from the perspective of this
literature is informative.  Like the standard lasso, singularity at
zero in $c(\varphi_{jk})$ causes some coefficients to be set to zero.
However, unlike the lasso, the gamma-lasso has gradient
$c'(\varphi_{jk}) = \mr{sign}(\varphi_{jk}) s/(r + |\varphi_{jk}|)$
which disappears as $|\varphi_{jk}| \rightarrow \infty$, leading to
the property of {\it unbiasedness for large coefficients} listed by
\cite{FanLi2001} and referred to as {\it Bayesian robustness} by
\citet{CarvPolsScot2010}.  Other results from this literature apply
directly; for example, in most problems it should be possible to
choose $s$ and $r$ to satisfy requirements for the strong oracle
property of \citet{FanPeng2004} conditional on their various
likelihood conditions.

It is important to emphasize that, despite sharing properties with
cost functions that are purpose-built to satisfy particular notions of
optimality, $c(\varphi_{jk})$ occurs simply as a consequence of proper
priors in a principled Bayesian model specification.  To illustrate
the effect of this penalty, Figure \ref{paths} shows MAP coefficients
for a simple logistic regression under changes to data and
parameterization.  In each case, gamma-lasso estimates threshold to
zero before jumping to solution paths that converge to the MLE with
increasing evidence.  Figure \ref{soln} illustrates
how these solution discontinuities arise due to concavity in
the minimization objective, an issue that is discussed in detail in
Section \ref{estimation}.3.  Note that although the univariate lasso
thresholds at larger values than the gamma-lasso, in practice we often
observe greater sparsity under GL penalties since large signals are
less biased and single coefficients are allowed to account for the
effect of multiple correlated inputs.  In contrast, standard lasso estimates
also fix some estimates at zero but lead to continuous solution paths
that never converge to the MLE.

\begin{figure}[p]
\hskip -.5cm \includegraphics[width=6.6in]{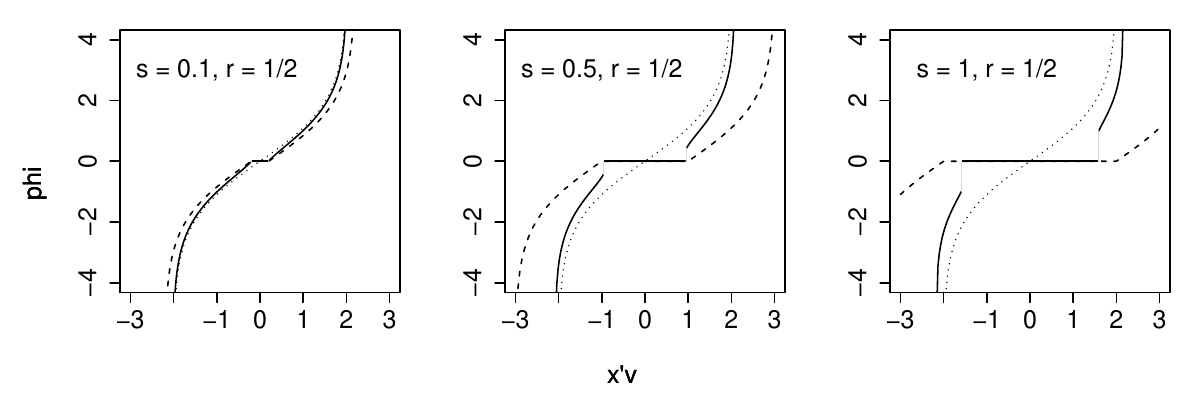}
\vskip -.5cm 
\caption{ \label{paths} Maximizing solutions for univariate logistic
  regression log posteriors $L(\varphi) = \bm{x}'\bm{v}\varphi -
  \sum_i \log\left[1+e^{\varphi v_i}\right] - \mr{pen}(\varphi)$, given
  $\bm{v} = [-1,-1,1,1]'$.  The dotted line is the MLE, with
  $\mr{pen}(\varphi)=0$, the dashed line is lasso $\mr{pen}(\varphi)=
  s|\varphi|/r$, and the solid line is gamma-lasso $\mr{pen}(\varphi) = s
  \log(1 + |\varphi|/r)$. }
\vskip 1cm
\hskip -.3cm \includegraphics[width=6.5in]{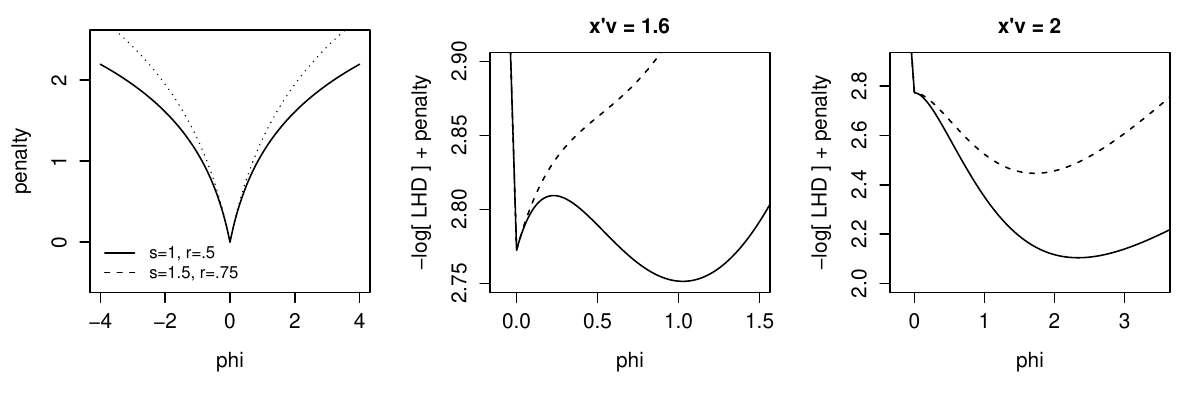}
\vskip -.5cm 
\caption{\label{soln} The left panel shows gamma-lasso penalty
  $s\log(1+|\varphi|/r)$ for $[s,r]$ of $[1,1/2]$ (solid) and
  $[3/2,3/4]$ (dashed).  The right two plots show the corresponding
  minimization objectives, negative log likelihood plus GL penalty,
  near a solution discontinuity in the simple logistic regression of Figure
  \ref{paths}.  }
\vskip 1cm
\includegraphics[width=6.3in]{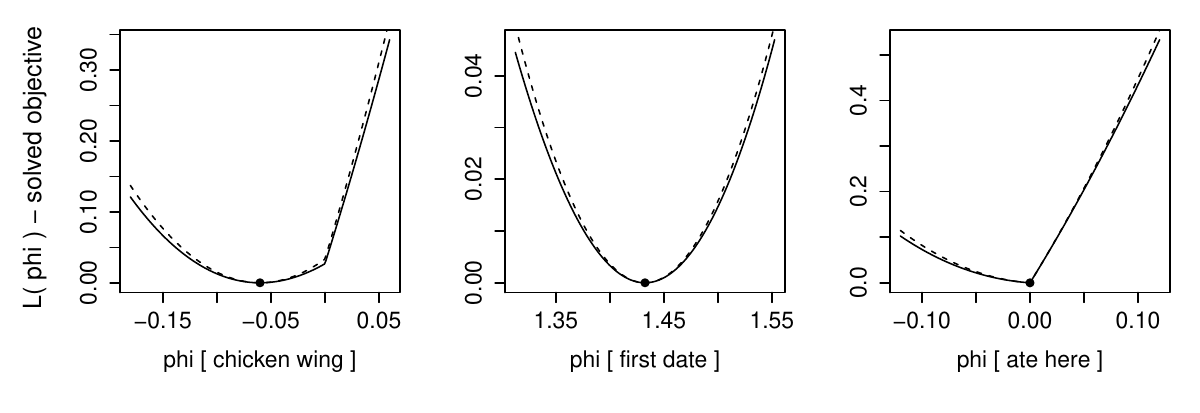}
\vskip -.4cm 
\caption{\label{coef} Coordinate objective functions at convergence in
  regression of we8there reviews onto overall rating. Solid lines are
  the true negative log likelihood and dashed lines are bound
  functions with $\delta = 0.1$.  Both are shown for
  new $\varphi_{j}^\star$ as a difference over the minimum at estimated
  $\varphi_{j}$ (marked with a dot).  }
\end{figure}

\subsection{Negative log posterior minimization by coordinate descent}

Taking negative log and removing constant factors,
maximization equates with minimization of
$l(\bs{\Phi},\bs{\alpha},\bm{U}) + \sum_{j=1}^p (\alpha_j/\sigma_\alpha)^2 -\log
\pi(\bm{U}) + c(\bm{\Phi})$, where $l$ is the strictly convex 
\begin{equation}\label{negllhd}
l(\bs{\Phi},\bs{\alpha},\bm{U}) = -\sum_{i=1}^n\left[\bm{x}_i'(\bs{\alpha} +
  \bs{\Phi}'\bm{v}_i + \bm{u}_i) - m_i\log\left(\sum_{j=1}^p \exp(\alpha_j +
    \bs{\varphi_j}'\bm{v}_i + u_{ij})\right) \right].
\end{equation}
Full parameter-set moves for this problem are prohibitively expensive
in high-dimension due to (typically dense) Hessian storage
requirements.  Hence, feasible algorithms make use of coordinate
descent (CD), wherein the optimization cycles through updates for each
parameter conditional on current estimates for all other parameters
\citep[e.g.,][]{Luen2008}.  Although conditional minima for logistic
regression are not available in closed-form, one can bound the CD
objectives with an easily solvable function and optimize that instead.
In such bound-optimization \citep[also known as
majorization;][]{LangHuntYang2000} for, say, $l(\theta)$, each move
$\theta^{t-1} \rightarrow \theta^t$ proceeds by setting new $\theta^t$
as the minimizing argument to bound $b(\theta)$, where $b$ is such
that previous estimate $\theta^{t-1}$ minimizes $b(\theta) -
l(\theta)$.  Algorithm monotonicity is then guaranteed through the
inequality $l(\theta^{t}) = b(\theta^{t}) + l(\theta^{t}) -
b(\theta^{t} ) \leq b(\theta^{t-1}) - \left[b(\theta^{t-1}) -
  l(\theta^{t-1})\right] = l(\theta^{t-1})$.

Using  $\theta^{\star}$ to denote a new value
for a parameter currently estimated at $\theta$, a quadratic bound for
each element of (\ref{negllhd}) conditional on all others is
available through Taylor expansion as
\begin{equation} \label{likbnd}
b(\theta^\star) = l(\bs{\Phi},\bs{\alpha},\bm{U}) +
g_l(\theta)(\theta^\star-\theta)
+ \frac{1}{2}(\theta^\star-\theta)^2 H_{\theta} 
\end{equation}
where $g_l(\theta) = \partial l /\partial \theta$ is the current coordinate
gradient and $H_{\theta}$ is an upper bound on curvature at the
updated estimate,
$h_l(\theta^\star) = \partial^2 l /\partial \theta^{\star 2}$.  Quadratic bounding
is also used in the logistic regression CD algorithms of
\cite{KrisCariFiguHart2005} and \citet{MadiGenkLewiFrad2005}: the
former makes use of a loose static bound on $h_l$, while the latter
updates $H_{\theta}$ after each iteration to obtain tighter
bounding in a constrained {\it trust-region} $\{\theta^\star \in
\theta \pm \delta\}$ for specified $\delta > 0$.  We have found that
dynamic trust region bounding can lead to an order-of-magnitude fewer
iterations, and Appendix A.2 derives $H_{\theta}$ as the least upper
bound on $h_l(\theta^\star)$ for $\theta^\star$ within $\delta$ of
$\theta$.

In implementing this approach,
coordinate-wise gradient and curvature for $\varphi_{jk}$ are
\vskip -1cm
\begin{equation}\label{lgrad}
g_l(\varphi_{jk}) = \frac{\partial l}{\partial \varphi_{jk}} = -\sum_{i=1}^n v_{ik} (x_{ij}
- m_i q_{ij}) ~~\text{and}~~ h_l(\varphi_{jk}) = \frac{\partial^2 l}{\partial
  \varphi_{jk}^2} = \sum_{i=1}^n m_iv_{ik}^2q_{ij}(1-q_{ij}),
\end{equation}
and similar functions hold for random effects and intercepts but with covariates
of one and without summing over $i$ for random effects.  Then under normal, say
$\mr{N}(\mu_\theta,\sigma_\theta^2)$, priors for $\theta = u_{ij}$ or
$\alpha_i$, the negative log posterior bound is $B(\theta^\star) =
b(\theta^\star) + 0.5(\theta -\mu_\theta)^2/\sigma_\theta^2$ which is
minimized in $\{\theta \pm \delta\}$ at $\theta^\star = \theta -
\mr{sgn}(\Delta \theta)\mr{min}\{|\Delta \theta|, \delta\}$ with
$\Delta \theta = \left[g_l(\theta)
  +(\theta-\mu_{\theta})/\sigma_\theta^2\right]/\left[H_{\theta} +
  1/\sigma_\theta^2\right]$.

Although the GL penalty on $\varphi_{jk}$ is concave and lacks a
derivative at zero, coordinate-wise updates are still available in
closed form. Suppressing the $jk$ subscript, each coefficient update
under GL penalty requires minimization of $ B(\varphi^\star) =
g_l(\varphi)(\varphi^\star-\varphi) +
\frac{1}{2}(\varphi^\star-\varphi)^2 H_{\varphi} + s\log(1 +
|\varphi^\star|/r) $ within the trust region $\left\{\varphi^\star \in
  \varphi \pm \delta : \mr{sgn}(\varphi^\star) =
  \mr{sgn}(\varphi)\right\}$.  This is achieved by finding the roots
of $B'(\varphi^\star) = 0 $ and, when necessary, comparing to the
bound evaluated at zero where $B'$ is undefined.  Setting $
B'(\varphi^\star) = 0$ yields the quadratic equation
\begin{equation}\label{quad}
\varphi^{\star 2} + \left( \mr{sgn}(\varphi)r - \tilde\varphi \right)
\varphi^{\star} + \frac{s}{H_{\varphi}} - \mr{sgn}(\varphi)r\tilde\varphi =0
\end{equation}
with characteristic $\left(\mr{sgn}(\varphi)r + \tilde\varphi\right)^2
- 4s/H_{\varphi} $, where $\tilde\varphi = \varphi -
g_l(\varphi)/H_{\varphi}$ would be the updated coordinate for an MLE
estimator.  From standard techniques, for $\{ \varphi^\star
: \mr{sgn}(\varphi) = \mr{sgn}(\varphi^\star)\}$ this function will
have at most one real minimizing root -- that is, with $H_{\varphi} >
s/\left(r+|\varphi^\star|\right)^2 $.  Hence, each coordinate update
is to find this root (if it exists) and compare $B(\varphi^\star)$ to
$B(0)$.  The minimizing value ($0$ or possible root $\varphi^\star$)
dictates our parameter move $\Delta \varphi$, and this move is
truncated at $\mr{sgn}(\Delta \varphi)\delta$ if it exceeds the trust
region.  Finally, when $\varphi = 0$, repeat this procedure for both
$\mr{sgn}(\varphi) =\pm 1$; at most one direction will lead to a
nonzero solution.

As it is inexpensive to characterize roots for $B'(\varphi^\star)$,
the gamma-lasso does not lead to any noticeable increase in
computation time over standard lasso algorithms
\citep[e.g.,][]{MadiGenkLewiFrad2005}.  Crucially, tests for decreased
objective can performed on the bound function, instead of the full
negative log posterior.  Figure \ref{coef} shows objective and bound
functions around the converged solution for three phrase loadings from
regression of we8there reviews onto overall rating.  With $\delta =
0.1$, $B$ provides tight bounding throughout this neighborhood.
Behavior around the origin is most interesting: the solution for {\it
  chicken wing}, a low-loading negative term, is at $B'(\varphi^\star)
= 0$ just left of the singularity at zero, while {\it ate
  here} falls in the sharp point at zero. The
neighborhood around {\it first date}, a high-loading term, is
everywhere smooth.

\subsection{Posterior log concavity and algorithm convergence}

Since the gamma-lasso penalty is everywhere concave, our minimization
objective is not guaranteed to be convex.  This is illustrated by the
right two plots of Figure \ref{soln}, where a very low-information
likelihood (four observations) can be combined with a relatively
diffuse prior on $\lambda$ ($s=1$, $r=1/2$) to yield concavity near
zero.  The effect of this is benign when the gradient is the same
direction on either side of the origin (as in the right panel of
\ref{soln}), but in other cases it will lead to local minima at zero away
from the true global solution (as in the center panel).  Such
non-convexity is the cause of the discontinuities in the
solution paths of Figure \ref{paths}.

From the second derivative of $l(\varphi_{jk}) + c(\varphi_{jk})$, the
conditional objective for $\varphi_{jk}$ will be concave only if
$h_l(\varphi_{jk} = 0) < s/r^2$ -- that is, if prior variance on $\lambda_{jk}$ is
greater than the negative log likelihood curvature at $\varphi_{jk} =
0$.  In our experience, this problem is rare: the likelihood typically
overwhelms penalty concavity and real examples behave like those shown
in Figure \ref{coef}.  Moreover, although it is possible to show
stationary limit points for CD on such nonconvex functions
\citep[e.g.][]{MazuFrieHast2011}, we advocate avoiding the issue
through prior specification.  In particular, hyperprior shape and rate
can be raised to decrease $\mr{var}(\lambda_{jk})$ while keeping
$\ds{E}[\lambda_{jk}]$ unchanged.  Although this may require more
prior information than desired, it is the amount necessary to have
both fast MAP estimation and estimator stability.  If you want to use
more diffuse priors, you should pay the computational price
of marginalization and mean inference \citep[as in,
e.g.,][]{GramPols2010}.

\section{Examples}
\label{results}

We now apply our framework to the datasets of
Section \ref{intro}.1. The implemented software is available as the
\code{textir} package for \code{R}, with these examples included as
demos.  Section \ref{results}.1 examines out-of-sample predictive
performance, and is followed by individual data analyses.

\subsection{A comparison of text regression methods}

Our prediction performance study considers three text analyses: both
constituent percentage vote-share for G.W. Bush (\code{bushvote}) and
Republican party membership (\code{gop}) regressed onto speech for a
member of the $109^\mr{th}$ US congress, and a user's overall rating
(\code{overall}) regressed onto the content of their we8there
restaurant review.  In each case, we report root mean square error or
misclassification rate over 100 training and validation iterations.
Full results and study details are provided in Appendix A.3, and
performance for a subset of models is plotted in Figure \ref{CV}.
Here, we focus on some main comparisons that can be drawn from the
study.

MNIR is considered under three different hyperprior specifications,
with rate $r=1/2$ and shapes of $s=1/100$, $1/10$, and $1$.  Response
factors are $v_i = y_i$ for \code{gop} and \code{overall}, and $v_i$
is set as $y_i$ rounded by whole number for \code{bushvote} (note that
instead setting $v_i = y_i$ here leads to no discernable improvement).
In each case, MNIR is fit for observations binned by factor level.  We
consider models both with and without independent random effects.  As
predicted, performance is unaffected by random effects for discrete
$y_i$, where we are collapsing together hundreds of observations.
However, they do improve out-of-sample performance by approximately
$1.5\%$ for \code{bushvote}, where only a small number of speakers are
binned at each whole percentage point.  Hence, detailed MNIR results
are reported {with random effects included only for \code{bushvote}.
  Finally, resulting SR scores $z_i = \bs{\varphi}' \bm{f}_i$ are
  incorporated into a variety of forward regression models: linear
  $\ds{E}[y_i] = \alpha + \beta z_i$ and quadratic $\ds{E}[y_i] =
  \alpha + \beta_1 z_i + \beta_2 z_i^2$ for \code{bushvote}, logistic
  $\ds{E}[y_i] = \exp[ \alpha + \beta z_i]/(1 + \exp[ \alpha + \beta
  z_i] )$ for \code{gop}, and linear and proportional-odds logistic
  $\mr{p}(y_i \leq c) = \exp[ \alpha_c - \beta z_i]/(1 + \exp[
  \alpha_c - \beta z_i] )$, $c=1\ldots5$, for \code{overall}.

  Performance is very robust to changes in the MNIR hyperprior.
  Figure \ref{CV} shows little difference between otherwise equivalent
  models using the conservative default $s=1$ and the lowest expected
  penalty $s=1/100$; results for $s=1/10$ are squeezed in-between.  In
  congressional speech examples $s=1/100$ has a slight edge; phrases
  here have already been pre-selected for partisanship and are thus
  largely relevant to the sentiment.  On the other hand, $s=1$ is the
  best performing shape for the we8there example, where phrases were
  only filtered by a minimum document threshold.  Looking at forward
  regressions, the problem specific quadratic \code{bushvote} (see
  Section \ref{results}.2 for justification) and proportional odds
  \code{overall} (accounting for ordinal response) forward regressions
  provide lower average out-of-sample error rates at the price of
  slightly higher variability across iterations, when compared to
  simple linear forward regression.

As comparators, we consider text-specific LDA (both supervised and
standard topic models) as well as an assortment of generic regression
techniques: lasso penalized linear (\code{bushvote} and
\code{overall}) and binary logistic (\code{gop})
regression, with penalty either optimized under our gamma hyperpriors
(\code{gop}), marginalized in MCMC (\code{bushvote}), or
tuned through CV (all examples); first-direction PLS
(\code{bushvote} and \code{overall}); and support vector
machines (\code{gop}).  In {\it every} comparison, gamma-lasso
MNIR provides higher quality predictions with lower run-times.  The
only similar predictive performance was for LDA with 25 and 50 topics
in the \code{bushvote} example, at 15-50 times higher
computational cost.  Note that, given the size of real text analysis
applications, we view the speed and scaleability of MNIR as a primary
strength and only considered feasible alternatives, with short Gibbs runs
for 50 topic sLDA and the Bayesian lasso (7-9 min) at the very high
end of our runtimes.  Moreover, both sLDA and CV lasso
occasionally fail to converge (these runs were excluded);
this never happened for MNIR.

Among comparators, the multinomial topic models outperform generic
alternatives.  Interestingly, LDA combined with simple regression
outperforms sLDA in both congress examples.  Again, this is probably
due to pre-selection of phrases: topics are relevant to ideology
regardless of supervision, and the extra parameters in sLDA are not
worth their cost in degrees of freedom.  Moreover, the simpler LDA
models can be fit with the MAP estimation of \citet{Tadd2012b},
whereas sLDA is applied here through a slow-to-converge Gibbs sampler
(we note that the original sLDA paper uses a variational EM
algorithm).  However, in the we8there data, the extra machinery of
sLDA offers a clear improvement over unsupervised LDA, as should be
the case in many text applications.  Finally, in an important side
comparison, binary logistic regressions were fit for \code{gop}
regressed onto phrase frequencies using both CV and independent gamma
hyperpriors for the lasso penalty.  The scaleable, low-cost,
gamma-lasso yields large performance improvements over a CV optimized
model, regardless of hyperprior specification.

\begin{figure}[p!]
\begin{center}
\vskip -1.5cm
\hspace{-.5cm}\includegraphics[width=6.5in]{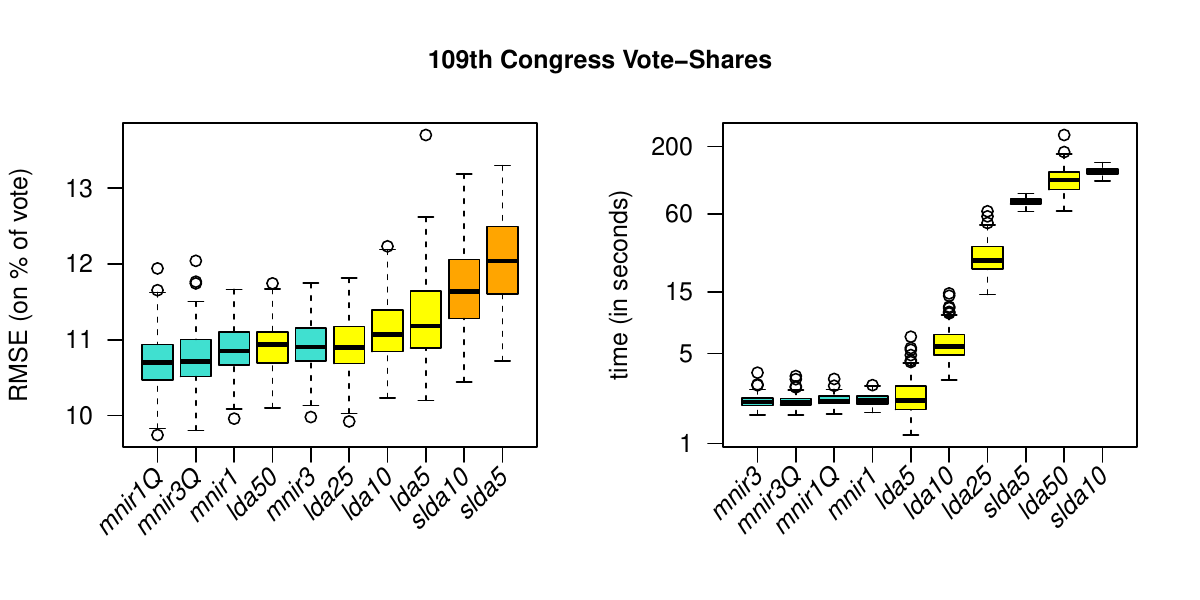}

\vskip -1.5cm
\hspace{-.5cm}\includegraphics[width=6.5in]{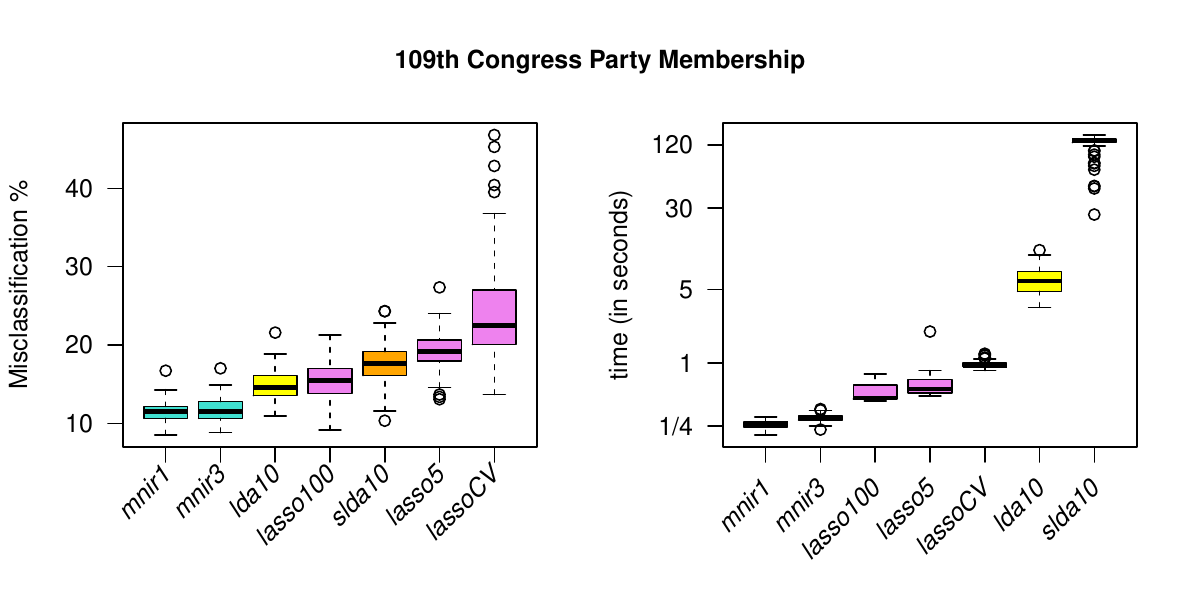}

\vskip -1.5cm
\hspace{-.5cm}\includegraphics[width=6.5in]{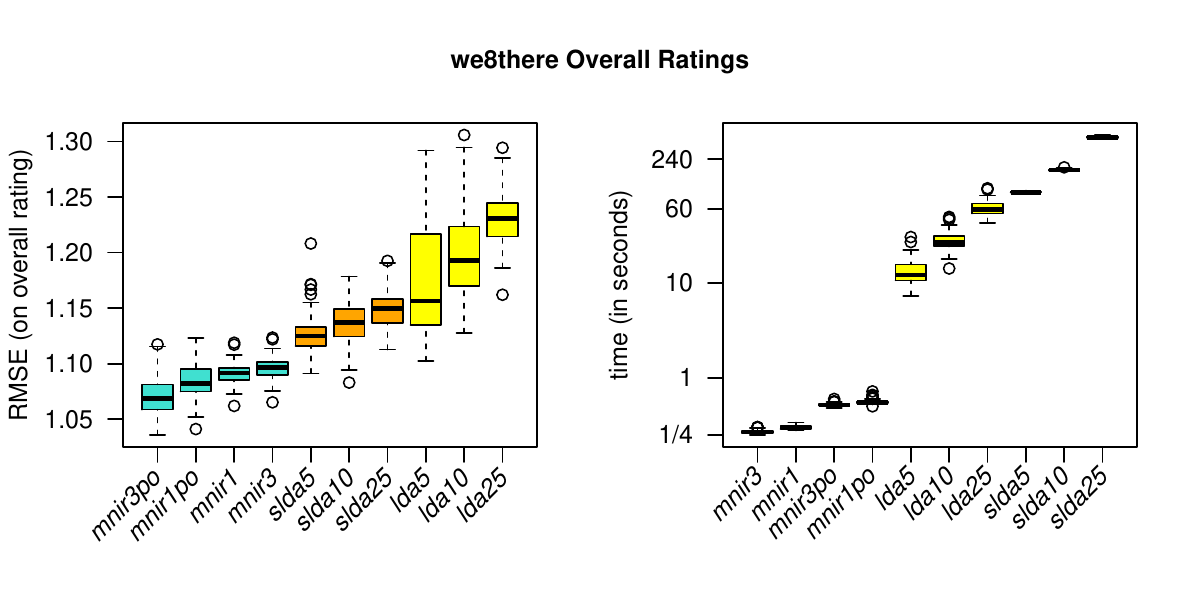}
\end{center}
\vskip -1cm
\caption{\label{CV} Out-of-sample performance and run-times for
  select models.  For MNIR, `Q' indicates quadratic and `po'
  proportional-odds logistic forward regressions, while $\lambda_j$
  prior `1' is $\mr{Ga}(0.01,0.5)$ and `3' is $\mr{Ga}(1,0.5)$.  We
  annotate with the number of topics for (s)LDA, and for binary Lasso
  regressions with either CV or the rate in an exponential penalty
  prior. Full details are in the appendix. }
\end{figure}

\subsection{Application: partisanship and ideology in political speeches}

For the data of Section \ref{intro}.1.1, we have two sentiment
metrics of interest: an indicator for party membership, and each
speaker's constituent vote-share for Bush in 2004.  Since the two
independents caucused with Democrats, the former metric can be
summarized in \code{gop} as a two-party {\it partisanship}.
Following the political economy notion that there should be little
discrepancy between voter and representative beliefs, {\sf\small
  bushvote} provides a measure of {\it ideology} as expressed in
support for G.W. Bush (and lack of support for John Kerry) in the
context of that election.

Figure \ref{sepfit} shows MNIR fit in separate models for each of
\code{gop} and \code{bushvote}, as studied in Section
\ref{results}.1.  For partisanship, fit with $s=1/100$ and $r=1/2$, a
simple univariate logistic forward regression yields clear
discrimination between parties; 8.5\% (45 speakers) are misclassified
under a maximum probability rule.  In the \code{bushvote} MNIR,
fit under the same hyperprior but with inclusion of random effects,
the resulting SR scores $z_i = \bs{\varphi}'\bm{f}_i$ increase quickly
with vote-share at low (mostly Democrat) values and more slowly for
high (mostly Republican) values.  This motivates our quadratic forward
regression for \code{bushvote} onto SR score, the predictive mean
of which is plotted in Figure \ref{sepfit} (with $R^2$ of 0.5).
However, looking at the SR scores colored by party (red for
Republicans, blue Democrats, green independents) shows that this
curvature could instead be explained through different forward
regression slopes by level of \code{gop}, implying that the
relationship between language and ideology is party-dependent.

Given the above, a more useful model might consider text reduction
that allows interaction between party and ideology.  For example, we
can build orthogonal bivariate sentiment factors as \code{gop} and
\code{bushvote} minus the \code{gop}-level means, say \code{votediff}
(again, rounded to the nearest whole percentage).  Figure \ref{bifit}
shows fitted values for such a model, including random effects
and with hyperprior shape increased to $s=1/10$ to reflect a 
preference for smaller conditional coefficients.  In detail, with
$z_{{\sf gop}}$ and $z_{{\sf votediff}}$ the two dimensions of SR
scores from MNIR $\bm{x} \sim \mr{MN}(\bm{q}(v_{{\sf gop}},v_{{\sf
    votediff}}), m)$, normalized for ease of interpretation, the
fitted forward model is
\begin{equation}
\ds{E}[{\sf bushvote}] = 51.9 + 6.2 z_{{\sf gop}} + 5.2 z_{{\sf  votediff}} - 1.9 z_{{\sf gop}} z_{{\sf  votediff}}.
\end{equation}
Thus a standard deviation increase in either SR direction implies a
5-6\% increase in expected vote-share, and each effect is dampened
when the normalized SR scores have the same sign.

\begin{figure}[p!]
\vspace{ -.9in}
\includegraphics[width=6.5in]{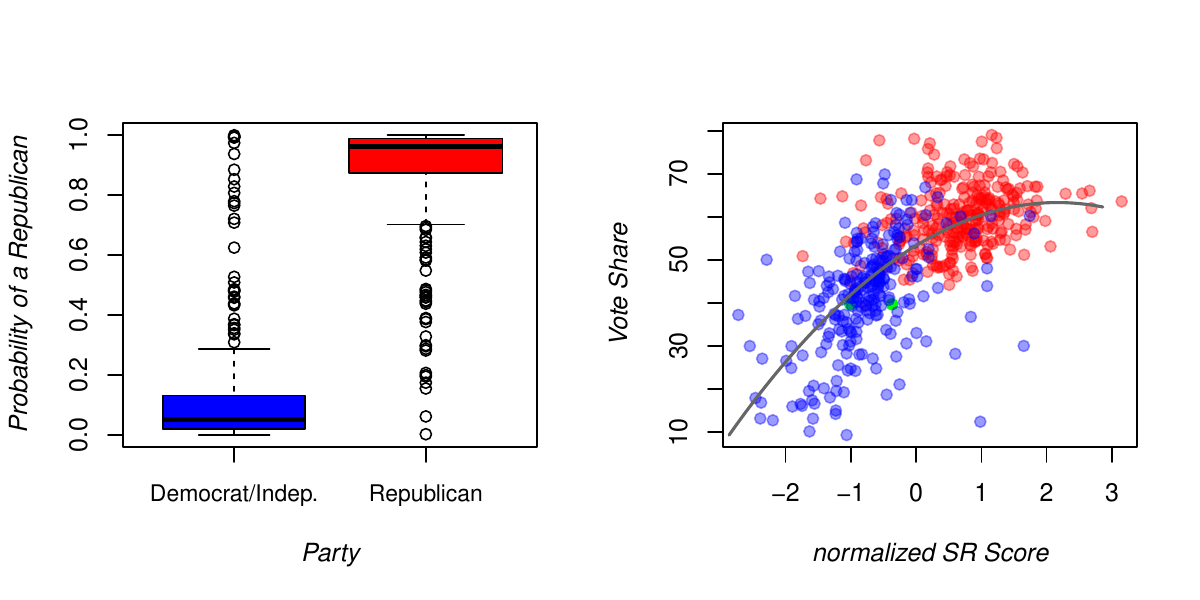}

\vspace{ -.15in}
\caption{\label{sepfit} Separate MNIR fits for congressional speech
  onto each of party and vote-share.  The right shows probabilities
  that each speaker is Republican and the left shows SR scores against
  \code{bushvote}. }

\vspace{ -.4in}
\includegraphics[width=6.5in]{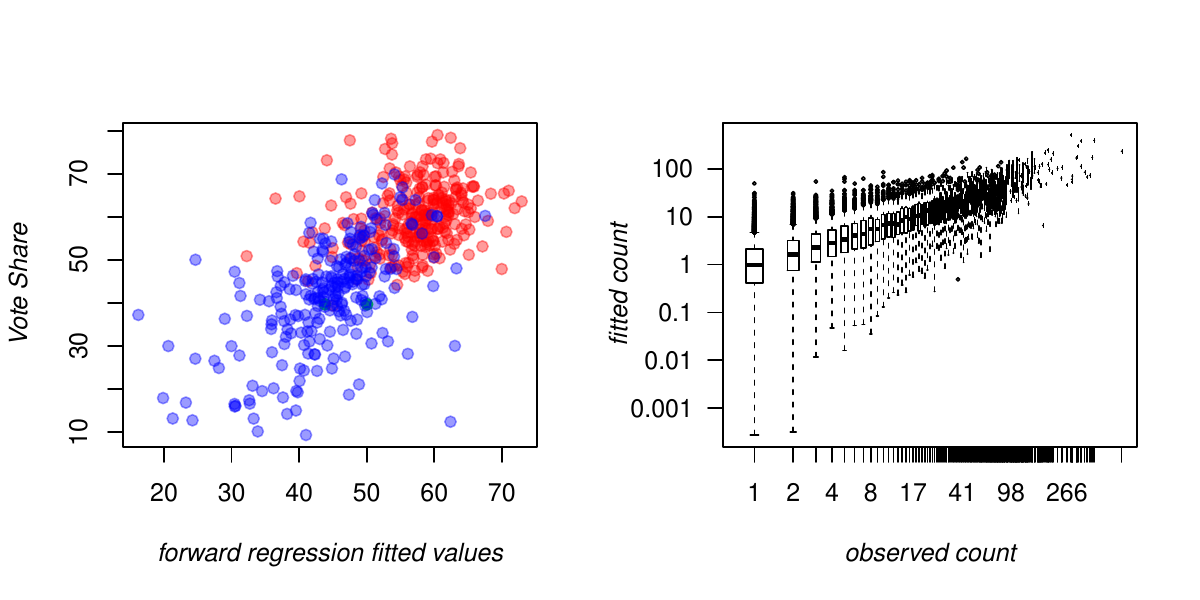}

\vspace{ -.15in}
\caption{\label{bifit} Bivariate ideology and partisanship MNIR.
  The left plot shows fitted values for a forward regression that
  interacts SR scores, and the right shows fitted vs observed token
  counts in MNIR. }

\vspace{ .1in}
\includegraphics[width=6.5in]{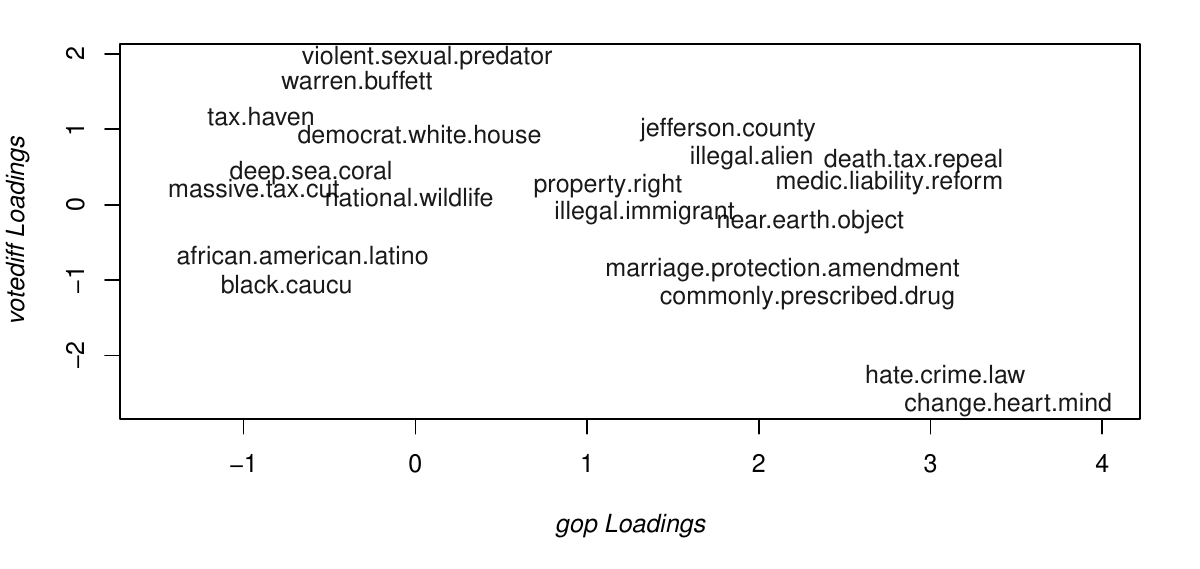}

\vspace{ -.2in}
\caption{\label{congressload} Select congressional speech term loadings
 in bivariate MNIR with party and vote-share.
}

\end{figure}

The right panel of Figure \ref{bifit} shows fitted expected counts
$q_j m$ against true nonzero counts in our bivariate MNIR model fit;
with random effects to account for model misspecification, there
appears to be no pattern of overdispersion.  The only clear outlier in
forward regression is Chaka Fattah (D-PA) with a standardized residual
of -5.2; he uttered a token in our sample only twice: once each for
\code{rate.return} and \code{billion.dollar}.  Finally, Figure
\ref{congressload} plots response factor loadings for a select group
of tokens.  Among other lessons, we see that racial identity rhetoric
(\code{african.american.latino}, \code{black.caucu}) points towards
the left wing of the Democratic party, while discussion of hate crimes
is indicative of a moderate Republican.  A few large loadings are
driven by single observations: for example,
\code{violent.sexual.predator} contributes more than 0.1\% of speech
for only Byron Dorgan, a Democratic Senator in Bush-supporting North
Dakota.  However, this is not the rule and most term loadings affect
many speakers.

\subsection{Application: on-line restaurant reviews}

\begin{figure}[p!]
\vspace{ -.75in}
\hskip -.5cm\includegraphics[width=6.9in]{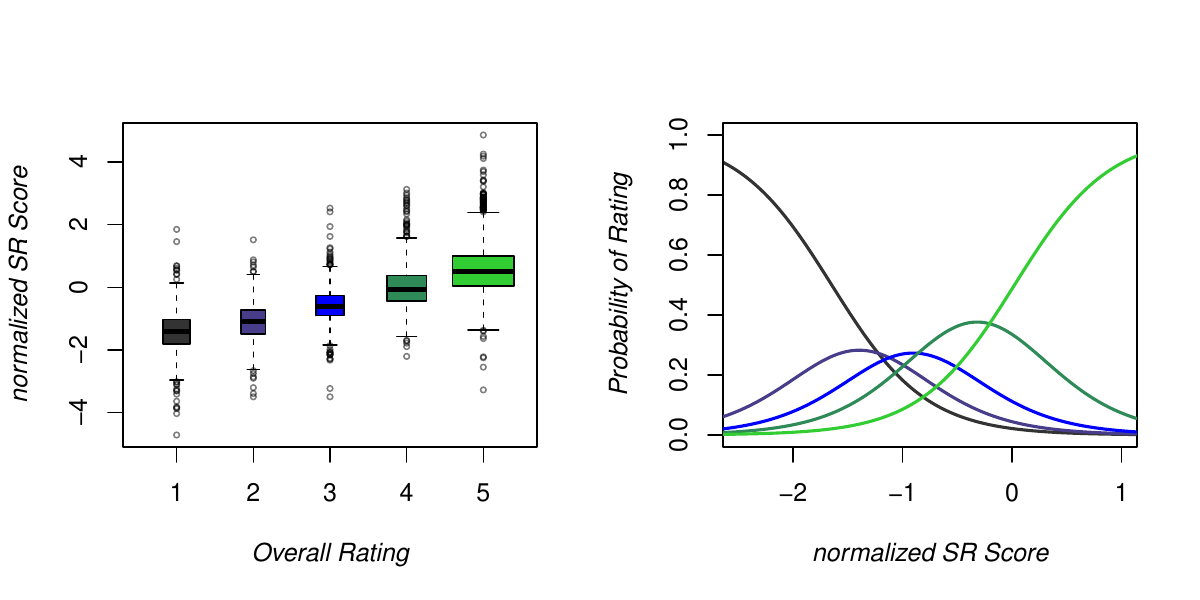}
\caption{\label{we8fit} Sufficient reduction and forward model fit
  for inverse regression of we8there reviews onto the corresponding
  overall rating.  The left plot shows SR score by true review rating, and
  the right shows proportional-odds logistic regression probabilities
  for each rating-level as a function of these SR scores.}

\vskip 1cm
\includegraphics[width=6.5in]{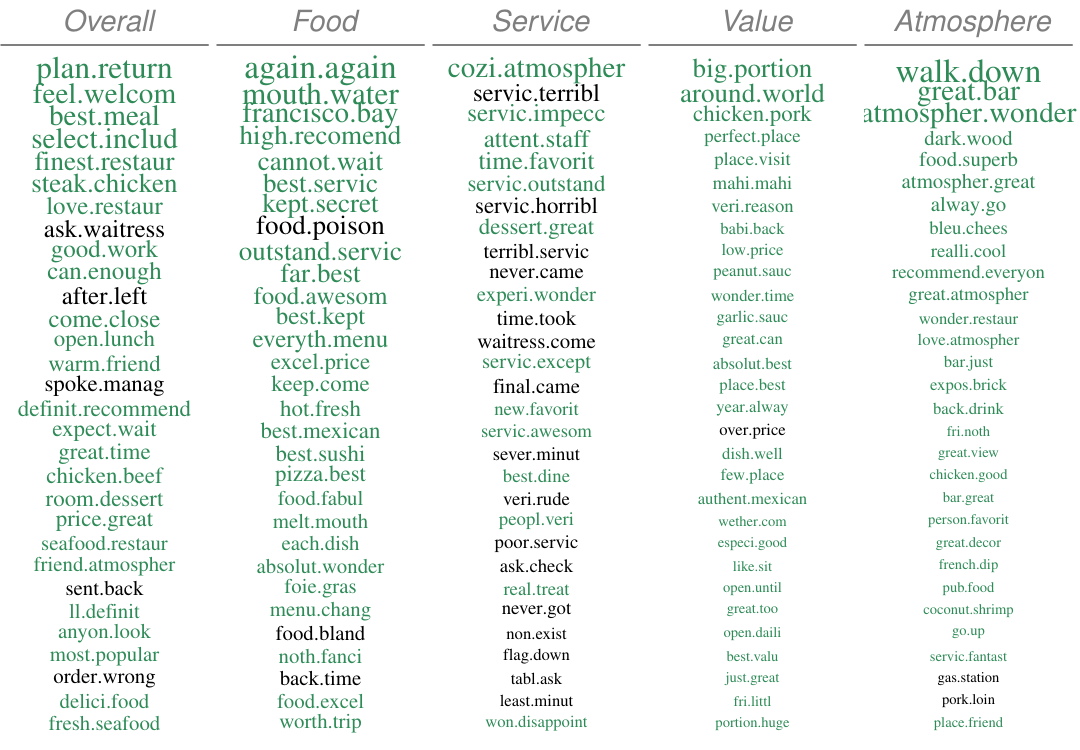}
\caption{\label{we8load} High-loading phrases in each direction
  for regression of we8there reviews onto aspect ratings.  Green
  tokens are
  positive, black are negative, and size is proportional to the absolute
  value of the loading.}
\end{figure}

For the data of Section \ref{intro}.1.2, our sentiment consists of
five correlated restaurant ratings (each on a five point scale) that
accompany every review.  The left panel of Figure \ref{we8fit} shows
MNIR for review content regressed onto the single \code{overall}
response factor, as studied in Section \ref{results}.1.  The true {\sf
  \small overall} rating has high correlation ($0.7$) with our SR
scores, despite considerable overlap between scores across rating
levels.  The right plot of Figure \ref{we8fit} shows probabilities for
each increasing \code{overall} rating category, as estimated in
the proportional-odds logistic forward regression, $\mr{p}({\sf
  overall} \leq c) = \exp[ \alpha_c - \beta z_{\sf overall}]/(1 +
\exp[ \alpha_c - \beta z_{\sf overall}] )$.  Again, $z_{\sf overall}$
is normalized here to have mean zero and standard deviation of one in
our sample.  This model has $\beta=2.3$, implying that the odds of
being at or above any given rating level are multiplied by $e^{2.3}
\approx 10$ for every standard deviation increase in the SR score.

Looking to explore aspect-specific factors, Figure \ref{we8load}
shows top-30 absolute value loadings in MNIR for review token-counts
onto {\it all five} dimensions of sentiment.  Influential terms on either
side of the rating spectrum can be easily connected with elements of a
good or bad meal: \code{plan.return}, \code{best.meal},
and \code{big.portion} are good, while \code{sent.back},
\code{servic.terribl}, and \code{food.bland} are bad.  The
largest loadings appear to be onto overall and food aspects, with
service slightly less important and loadings for value and atmosphere
quickly decreasing in size. This would indicate that the reviews focus
on these elements in that order.

\section{Discussion}
\label{discuss}

The promising results of Section \ref{results} reinforce a basic idea:
a workable inverse specification can introduce
information that leads to more efficient estimation.  Given the
multinomial model as a natural inverse distribution for token-counts,
analysis of sentiment in text presents an ideal setting for inverse
regression.  While the approach of not {\it jointly} modeling a
corresponding forward regression falls short of full Bayesian
analysis, such inference would significantly complicate estimation and
detract from our goal of providing a fast default method for
supervised document reduction.  We are happy to take advantage of
parametric hierarchical Bayesian inference for the difficult MNIR
estimation problem, and suggest that application appropriate
techniques for low-dimensional forward regression should be readily
available.

Although the illustrative applications in this article are quite
simple, the methods scale to far larger datasets.  Collapsing
observations across sentiment factors for MNIR yields massive
computational gains: training data need only include token counts
tabled by sentiment level, and as an example, in \citet{Tadd2012c}
this allows MNIR runs of only a few seconds for 1.6 million twitter
posts scored as positive or negative.  Moreover, we see no reason why
gamma-lasso logistic regression, which was developed specifically for
large response settings, should not be viewed as an efficient option
in generic penalized regression.  Finally, current collaborations that
use MNIR for text analysis include study of partisanship in the US
congressional record from 1873 to present, and an attempt to quantify
the economic content of news in 20 years of Wall Street Journal
editions.  In each case, we are considering a more rigorous treatment
of the identification of single sentiment dimensions and controlling
for related endogenous variables; this work shows MNIR's promise as
the basis for a variety of text related inference goals.

\newpage
\appendix
\section*{Appendix}
\setstretch{1.2}

\subsection*{A.1 Slant and Partial Least Squares}

The GS slant index for document $i$ is $z^\mr{slant}_{i} =
\sum_{j=1}^p b_{j}(f_{ij}-a_j)/\sum_{j=1}^p b_{j}^2$, with parameters
obtained through ordinary least-squares (OLS) as $[a_j,b_j] = \arg
\min_{a,b} \sum_{i=1}^n [f_{ij} - (a +b y_i)]^2$ for $j=0\ldots p$.
Since $b_j = \mr{cov}(f_j,y)/\mr{var}(y)$, slant is equivalent (up to
a uniform shift and scale for all index values) to a weighted sum of
term frequencies loaded by their covariance with $y$.  This is also
the first direction in partial least-squares; see \citet{FranFrie1993}
for statistical properties of PLS and its relationship to OLS, and
\cite{HastTibsFrie2009} for a common version of the algorithm.  Using
the usual normalization applied in PLS, an improved slant measure is
given by $z^\mr{slant}_{i} = \sum_{j=1}^p f_{ij}\mr{cor}(f_j, y_i)$.
For vote-share regressed onto congressional speech in the data of
Section \ref{intro}.1.1, this change
increases within-sample $R^2$ from $0.37$ to $0.57$.

Given $\bm{\hat F} = [ \bm{\hat f}_1 \cdots  \bm{\hat f}_p]$ as a normalized covariate matrix with
mean-zero and variance-one columns, a PLS algorithm which highlights
its inverse regression structure is as follows.
\begin{enumerate} 
\item Set the initial response factor $\bm{v}_0 = \bm{y} = [y_1 \ldots
  y_n]'$, and for $k=1,\ldots,K$:
\begin{itemize}\vspace{-.1cm}\sgl
\item[-] Loadings are $\bs{\varphi}_k =\mr{cor}( \bm{\hat{F}},
  \bm{v}_{k-1}) = [\mr{cor}(\bm{\hat f}_1, \bm{v}_{k-1}) \ldots
  \mr{cor}(\bm{\hat f}_p, \bm{v}_{k-1})]'$.
\item[-] The $k^{th}$ PLS direction
  is $\bm{z}_k = \bs{\varphi}_k'\bm{\hat{F}}$.
\item[-] The new response factors are $\bm{v}_k = \bm{v}_{k-1} -
  [\bm{z}_k'\bm{v}_{k-1}/(\bm{z}_k'\bm{z}_k)]\bm{z}_k$.
\end{itemize}
\vspace{-.3cm}
\item Set $\bm{\hat{y}}$ as OLS fitted values for
  regression of $\bm{y}$ onto $\bm{Z}$, where $\bm{Z} = [\bm{z}_{1}\cdots \bm{z}_{K}]$.
\end{enumerate}
An extra step to normalize and orthogonalize $\bm{z}_k$ with respect
to $[\bm{z}_1 \cdots \bm{z}_{k-1}]$ recovers orthonormal directions, as in
the original PLS algorithm.  Moreover, loading calculations replaced
by $\varphi_{kj} = \arg \min_{\varphi} \sum_{i=1}^n [f_{ij} - (a
+\varphi v_{ki})]^2$ will only scale $\bm{z}_k$ by the variance of
$\bm{v}_k$ and lead to the same forward fit, such that PLS can be
viewed as stagewise inverse regression.

\subsection*{A.2 Trust-region bound for logistic multinomial likelihood }

The bounding used here is essentially the same as in
\citet{GenkLewiMadi2007} but for introduction of dependence upon
$v_{ik}$ that is missing from their version.  
We describe the bound for updates to $\varphi_{jk}$, but it applies directly to $\alpha_j$
or $u_{ij}$ simply by replacing covariate values with one.

Given a trust region of $\varphi_{jk} \pm \delta$, the upper bound on
$h_l(\varphi_{jk}) = \sum_{i=1}^n v_{ik}^2m_i q_{ij}(1-q_{ij})$ is
$H_{jk} = \sum_{i=1}^n v_{ik}^2m_i/F_{ij}$, where each $F_{ij}$ is a
lower bound on $1/(q_{ij} -q_{ij}^2) = 2 + e^{\eta_{ij}+\delta v_{ik}}/E_{ij}
+ E_{ij}/e^{\eta_{ij} +\delta v_{ik}}$, with $E_{ij} = \sum_{l=1}^p
e^{\eta_{il}} - e^{\eta_{ij}}$. This target is
convex in $\delta$ with minimum at $e^{\delta v_{ik}} = 
  E_{ij}/e^{\eta_{ij}}$, such that

 \vspace{-.75cm} 
\begin{equation*}
F_{ij} = \frac{e_{ij}}{E_{ij}} +
\frac{E_{ij}}{e_{ij}} + 2~\text{where}~~e_{ij} = \left\{ 
\begin{array}{cc}
\!\!e^{\eta_{ij}-|v_{ik}|\delta} &\!\!\text{if} ~~E_{ij} < e^{\eta_{ij}-|v_{ik}|\delta}\\
\!\!e^{\eta_{ij}+|v_{ik}|\delta} &\!\!\text{if}~~E_{ij} >  e^{\eta_{ij}+|v_{ik}|\delta}\\
\!\!E_{ij} &\!\!\text{otherwise}.
\end{array}\right.
\end{equation*}
We use unique $\delta_{jk}$ and update
$\delta_{jk}^\star = \mr{max}\{ \delta_{jk}/2, 2|\varphi_{jk}^\star -
\varphi_{jk}|\}$ after each iteration.

\newpage
\hspace{-1in}
\begin{sideways}
\parbox{9.5in}{
\setstretch{1}
\subsection*{A.3 Out-of-Sample Prediction Study Details}
Each model was fit to 100 random data subsets and used to predict on the left-out sample.  Tables report
average root mean square error (RMSE) or percent misclassified (MC\%), the percentage
worse than best on this metric, and run-time in seconds (including
count collapsing in MNIR).  

\vskip .2cm\noindent We use {\sf\smaller R} package implementations:
{\sf\smaller lda} for SLDA \citep{Chan2011}; {\sf\smaller glmnet} for
CV lasso regression \citep{FrieHastTibs2010}; {\sf\smaller monomvn}
for Bayesian lasso \citep{Gram2012}; {\sf\smaller kernlab} for SVM
\citep{KaraSmolaHornZeil2004}; {\sf\smaller textir} for MNIR, LDA,
PLS, and gamma-lasso regression; and {\sf\smaller arm}
\citep{GelmSuYajiHillPittKermZhen2012} for the forward
regression models that accompany MNIR and LDA.  Penalty prior in MNIR
is $\mr{Ga}(s, 1/2)$, (s)LDA Dirichlet precisions are
$1/K$ for topic weights and $1/p$ for token probabilities, and
sLDA assumes a forward error variance of 25\% of marginal
response variance.  Unless otherwise specified, we apply
package defaults.  (S)LDA and MNIR use token counts; all others
regress onto token frequencies.

\vskip .4cm
\noindent {\it Vote Share:} Congressional speech with
two-party vote share (\%) as continuous response, training on 200
and predicting on 329.  Constant mean RMSE is 13.4. MNIR models were fit {\it with} random
effects;  models without random effects are an average of
1.5\% worse on RMSE but 20\% faster.  Bayes lasso uses a $\mr{Ga}$(2,1/10)
prior on $\lambda$ and was run for 200 MCMC iterations after a burn-in
of 100 (refer to 
{\sf\smaller monomvn} for details).

\begin{center}
{\footnotesize 
\begin{tabular}{l|ccc|ccc|ccccc|ccccc|cc|c}
\multicolumn{1}{c}{}
&\multicolumn{3}{c}{MNIR \& Quadratic} 
&\multicolumn{3}{c}{MNIR \& Linear}
&\multicolumn{5}{c}{ LDA \& Linear}
&\multicolumn{5}{c}{ Supervised LDA }
& \multicolumn{2}{c}{ Lasso } &\multicolumn{1}{c}{PLS}\\
& $s$ = $10^{-2}$&$10^{-1}$&$1$ 
& $s$ = $10^{-2}$&$10^{-1}$&$1$ 
& $K$ = 2&5&10&25&50
& $K$ = 2&5&10&25&50
& CV&Bayes&K=1\\
\hline
RMSE&10.7&10.7&10.8&10.9&10.9&10.9&11.7&11.3&11.1&10.9&10.9&12.9&12.1&11.7&12.3&15.1 &13.7&15.7&15.9 \\
\% Worse&0&0&0& 1&1&2&9&6&4&2&2&21&13&9&15&41&28&46&49 \\
Run Time&2.2&2.3&2.1&2.2&2.3&2.1&1.2&2.4 &6.2&29&112&43&75&128&288&508&0.9&410&0.1 \\
\hline
\end{tabular} }
\end{center}

\vskip .2cm
\noindent 
{\it Party Classification:} Congressional speech data with
`Republican' as binary response, training on 200
and predicting on 329.  Null model misclassification rate is 46\%.
MNIR models were fit {\it without} random
effects which lead to the same
misclassification but 40\% longer average run-times.    Lasso and
gamma-lasso are applied in binary logistic
regressions, with shape one and rate $r$ for the latter, and  SVM uses Gaussian kernels
with misclassification cost $C$ (refer to {\sf\smaller
  kernlab} for details).  LDA led to complete separation and
SLDA failed to converge for $K>10$.

\begin{center}
{\footnotesize 
\begin{tabular}{l|ccc|ccc|ccc|c|cccc|ccc}
\multicolumn{1}{c}{}
&\multicolumn{3}{c}{MNIR \& Logistic} 
&\multicolumn{3}{c}{LDA \& Logistic}
&\multicolumn{3}{c}{ Supervised LDA }&\multicolumn{1}{c}{Lasso
}
& \multicolumn{4}{c}{ Gamma-Lasso} &
\multicolumn{3}{c}{ SVM }
\\
& $s$ = $10^{-2}$&$10^{-1}$&$1$ 
& $K$ = 2&5&10
& $K$ = 2&5&10
&CV &$r$ = 5&25&50&100
&$C$=1&100&1000\\
\hline 
MC\%&11&11&12&20&15&15&33&20&18&24&19&17&16&15&37&32&32\\
\%  Worse&0&0&2&76&36&30&188&75&54&115&68&49&42&35&224&182&180\\
Run Time &0.3&0.4&0.3&1.1&2.5&6.3&44&77&126&1.0&0.6&0.5&0.5&0.5&3.1&3.5&3.4 \\
\hline
\end{tabular} }
\end{center}

\vskip .2cm
\noindent 
{\it Restaurant Rating:} We8there reviews with ordinal rating response, training on 2000
and predicting on 4166.  Constant mean RMSE is 1.35. Reported MNIR models 
were fit {\it without} random
effects which lead to equivalent
predictive performance but 15\% longer average run-times.  

\begin{center}
{\footnotesize 
\begin{tabular}{l|ccc|ccc|ccccc|ccccc|c|c}
\multicolumn{1}{c}{}
&\multicolumn{3}{c}{MNIR \& POLR } 
&\multicolumn{3}{c}{MNIR \& Linear }
&\multicolumn{5}{c}{ LDA \& POLR }
&\multicolumn{5}{c}{ Supervised LDA }
& \multicolumn{1}{c}{Lasso}&\multicolumn{1}{c}{PLS}\\
& $s$ = $10^{-2}$&$10^{-1}$&$1$ 
& $s$ = $10^{-2}$&$10^{-1}$&$1$ 
& $K$ = 2&5&10&25&50
& $K$ = 2&5&10&25&50
& CV&$K=1$\\
\hline
RMSE&1.08&1.08& 1.07&1.09&1.09&1.10&1.19&1.17&1.20&1.23&1.23&1.15&1.13&1.14&1.15&1.16&1.24&1.25  \\
\% Worse&1&1&0&2&2&2&12&10&12&15&15&8&5&6&7&8&16&17 \\
Run Time&0.6&0.6&0.5&0.3&0.4&0.3&2.5&13.4&28&61&167&53&90&154&341&651&54&2.2  \\
\hline
\end{tabular} }
\end{center}

}
\thispagestyle{empty}
\end{sideways}

\setstretch{1}\small
\bibliographystyle{chicago}
\bibliography{taddy}

\end{document}